\documentclass[11pt,letterpaper]{article}
\usepackage{amsmath,amsfonts,amssymb,amsthm}
\usepackage[top=2cm,bottom=2cm,left=2.5cm,right=2.5cm]{geometry}
\usepackage{graphicx}
\usepackage{subcaption}

\newtheorem{theorem}{Theorem}

\newcommand{\aLm}{\tfrac{\alpha\Lambda}{\mu}}
\newcommand{\bLm}{\tfrac{b\Lambda}{\mu}}

\usepackage[backend=bibtex,firstinits=true]{biblatex}
\bibliography{Covid-Pneumonia-coinfection-model}
\renewbibmacro{in:}{}

\usepackage{hyperref}

\title{
	A model for COVID-19 and bacterial pneumonia coinfection with community- and hospital-acquired infections
	}

\author{
	Angel G. C. P\'erez$^\text{1}$
	\and    
	David A. Oluyori$^\text{2}$
}

\date{}

\begin{document}

\maketitle

\begin{center}
	{\small
		$^\text{1}$
		Facultad de Matem\'aticas, Universidad Aut\'onoma de Yucat\'an, M\'erida, Yucat\'an, Mexico.
		Email address: \url{agcp26@hotmail.com}\\
		
		$^\text{2}$
		Department of Mathematics, School of Physical Science, Ahmadu Bello University, Zaria, Kaduna State, Nigeria.
		Email address: \url{oluyoridavid@gmail.com}
	}
\end{center}
\bigskip

\begin{abstract}
	We propose a new mathematical model to study the coinfection dynamics of COVID-19 and bacterial pneumonia. Our model includes two infection ways for pneumonia, corresponding to community-acquired and hospital-acquired infections. We show that the existence and local stability of equilibria depend on three different parameters, which are interpreted as the basic reproduction numbers of COVID-19, bacterial pneumonia, and bacterial population. Numerical simulations are performed to complement our theoretical analysis, and we show that both diseases can persist if the basic reproduction number of COVID-19 is greater than one.
\end{abstract}

\section{Introduction}

The Coronavirus Disease 2019 (COVID-19) has been a major public health concern across the nations of the world since its declaration as a global pandemic due to its rapid infectivity and  high death toll. Trend analysis reveals that one major cause of death due to Coronavirus has been secondary causes due to bacterial and viral infections, which lead to the eventual death. In the current realities from COVID-19, many studies have shown that Respiratory Tract Infections (RTIs) predispose patients to coinfections, which result in increased disease severity and death. RTIs are infections of the parts of the body involved in breathing, such as sinuses, throat, airways or lungs, caused by a variety of bacteria and virues such as Influenza (flu). Typical infections of the upper respiratory tract include tonsillitis, pharyngitis, sinusitis and certain types of influenza (such as H1N1). Symptoms of RTIs include cough, soar throat, running nose, nasal congestion, headache, low-grade fever, facial pressure and sneezing.

As reported in \cite{morens2008predominant}, most of the fatalities in the 1918 Influenza pandemic were due to subsequent bacterial infection, particularly with \textit{Streptococcus pneumoniae}. Data evidence from few studies has it that poor outcomes in the Influenza (H1N1) pandemic were associated with coinfections \cite{macintyre2018role}. So far, coinfections are increasingly recognised in respiratory tract infections such as MERS, SARS-CoV2, Influenza (H1N1) with the discovery of highly sensitive techniques for microorganism detection and identification (MALDI-TOF, Multiplex PCR). The study of coinfections in a pandemic situation such as COVID-19 has become imperative due to the clinical, diagnostic and therapeutic challenges it poses \cite{kumar2020covid}. To further buttress the aforestated, Lansbury et al. \cite{lansbury2020coinfections} highlighted some important aspects of bacterial and viral infections in COVID-19 and antimicrobial prescription.

Despite proven epidemiological significance of coinfections in the severity of respiratory diseases, they are largely understudied during a large outbreak of respiratory infections like SARS-CoV-2 \cite{cox2020coinfections}. Zhou et al. \cite{zhou2020clinical} showed that 50\% of the mortalities due to COVID-19 result from secondary bacterial infections. Chen et al. \cite{chen2020epidemiological} in this vein reported both bacterial and fungal infection. Clinical evidences shows that diagnosing coinfections is a complex process, because the organism itself might have been resident in the host before the viral infection as part of an underlying chronic infection or might have been contacted nosocomially \cite{cox2020coinfections}. Hence, early diagnosis of coinfections is required, preferably using a broad potential pathogens and antimicrobial resistances with subsequent monitoring for infection development. Therefore, to accurately diagnose and study coinfections in COVID-19, it is highly recommended that patients must be recruited on admission to intensive care units (ICU) and sampled longitudinally throughout the disease course using culture-independent techniques capable of identifying complex mixed infections without previous target selections such as whole-genome metagenomics to help identify the pathogens and make informed antibiotics prescription. As rapid extension of coinfection is necessary in the management and treatment of most severe COVID-19 cases, which could help save lives and improve antimicrobial stewardship. It has been reported that some patients presenting to the hospital with SARS-CoV-2 infection have a clinical phenotype that is not dissimilar from atypical bacterial pneumonia \cite{rawson2020bacterial}.

Recent studies have established clinical evidences of coinfections of SARS-CoV-2 (COVID-19) with other diseases such as tuberculosis \cite{kumar2020covid,orozco2020covid,tadolini2020active,yadav2020case,khurana2020significance,petrone2021coinfection,tadolini2020active}, influenza A (H1N1) \cite{lew2020coinfection,jing2021coinfection,fahim2021coinfection,xiang2021coinfection,ata2020year} and Middle East Respiratory Syndrome Coronavirus (MERS-CoV) \cite{elhazmi2021severe}, as well as bacterial coinfections \cite{giannella2022predictive}. Due to this, some authors have developed mathematical models to study the dynamics of COVID-19 and its coinfection with influenza \cite{soni2021covid}, malaria \cite{tchoumi2021malaria}, tuberculosis \cite{bandekar2022coinfection,rwezaura2022mathematical}, dengue \cite{omame2021covid} and diabetes \cite{omame2022diabetes}. However, no model has been proposed to study the coinfection dynamics of COVID-19 with bacterial pneumonia.

Bacterial pneumonia is an inflammation of the lungs caused by infection with certain bacteria. Depending on the location where a person acquires the infection, it can be classified as either \textit{community-acquired pneumonia} or \textit{hospital-acquired pneumonia}. Community-acquired pneumonia is by far the most common type \cite{depietro2019what}. On the other hand, hospital-acquired pneumonia is usually more severe because the infecting organisms tend to be more aggressive; they are also less likely to respond to antibiotics and are, therefore, harder to treat \cite{setih2020hospital}. Clinical studies have shown that critically ill COVID-19 patients admitted to the hospital suffer more frequent bacterial or fungal nosocomial infections, and patients with underlying risk factors such as advanced age, mechanical ventilation or prolonged hospital stay are more prone to these complications \cite{ansari2021potential,huang2020clinical,wang2020clinical}. Bacterial or fungal coinfections are unlikely to be common in patients with mild COVID-19 when compared with those with more severe disease upon admission to the hospital \cite{ansari2021potential}.

The present study is motivated by the need to mathematically study the dynamics of coinfection of COVID-19 with bacterial pneumonia, including the cases when bacterial infection is acquired in the community or in the hospital. This paper is structured as follows: in Section \ref{sec:description}, we introduce three models: a sub-model for COVID-19 infection, a sub-model for bacterial pneumonia, and coinfection model that includes the dynamics of both diseases. In Section \ref{sec:sub-models}, we determine some basic properties for the two sub-models. In Section \ref{sec:coinfection}, we provide an analysis for the coinfection model. In Section \ref{sec:num}, we perform some numerical simulations to illustrate the dynamics of the coinfection model. Finally, in Section \ref{sec:conclusions}, we provide some concluding remarks.

\section{Description of the models}
\label{sec:description}

\subsection{COVID-19 infection model}

The COVID-19 infection model subdivides the human population into four compartments: susceptible ($S$), infected but not hospitalised ($I$), hospitalised ($H$), and recovered ($R$). This model can be described by the following system of equations:
\begin{equation}\label{COVID-model}
\begin{aligned}
	S' & = \Lambda + \sigma R - \mu S - \alpha SI, \\
	I' & = \alpha SI - (\gamma + \eta + \mu)I,     \\
	H' & = \eta I - (\theta + \delta + \mu)H,      \\
	R' & = \gamma I + \theta H - \mu R - \sigma R.
\end{aligned}
\end{equation}

The interpretation of parameters is as follows:
\begin{itemize}
	\item $\Lambda$: recruitment rate of susceptible population.
	\item $\mu$: natural death rate.
	\item $\alpha$: transmission rate of COVID-19.
	\item $\gamma$: recovery rate of people infected with COVID-19 but not hospitalised.
	\item $\theta$: recovery rate of hospitalised people.
	\item $\eta$: hospitalisation rate.
	\item $\delta$: COVID-19-induced death rate of hospitalised people.
	\item $\sigma$: rate of loss of immunity against COVID-19 infection.
\end{itemize}

\subsection{Bacterial pneumonia infection model}

The model for bacterial pneumonia subdivides the human population into three compartments: susceptible ($S$), infected ($I$), and recovered ($R$). We also consider a compartment $B$ representing the population of bacteria in the environment. The model is given by the following system:
\begin{equation}\label{bacterial-model}
\begin{aligned}
	S' & = \Lambda - \mu S - bSI - b_1SB,               \\
	I' & = bSI + b_1SB - \phi I - \mu I - \delta I,    \\
	R' & = \phi I - \mu R,                              \\
	B' & = pI + rB\left(1-\frac{B}{\kappa}\right) - mB.
\end{aligned}
\end{equation}

The parameters of this model can be interpreted as follows:
\begin{itemize}
	\item $\Lambda$: recruitment rate of susceptible population.
	\item $\mu$: natural death rate.
	\item $b$: transmission rate of community-acquired bacterial pneumonia.
	\item $b_1$: transmission rate of hospital-acquired bacterial pneumonia.
	\item $\delta$: disease-induced death rate of infected population.
	\item $\phi$: recovery rate of people with bacterial infection.
	\item $p$: rate of excretion of bacteria in the environment by infected people.
	\item $r$: maximal per capita growth rate of bacteria in the environment.
	\item $\kappa$: carrying capacity of bacterial population.
	\item $m$: clearance rate of bacterial population.
\end{itemize}

\subsection{Coinfection model}

Based on models \eqref{COVID-model} and \eqref{bacterial-model}, we propose a combined COVID-19--bacterial pneumonia coinfection model. We will consider three stages for COVID-19 infection and four for bacterial infection, which gives twelve mutually exclusive compartments: bacterial pneumonia susceptible and COVID-19 susceptible ($X_{SS}$); bacterial pneumonia susceptible and COVID-19 mildly infected ($X_{SI}$); bacterial pneumonia susceptible and COVID-19 hospitalised ($X_{SH}$); bacterial pneumonia susceptible and COVID-19 recovered ($X_{SR}$); bacterial pneumonia infected and COVID-19 susceptible ($X_{IS}$); bacterial pneumonia infected and COVID-19 mildly infected ($X_{II}$); bacterial pneumonia infected and COVID-19 hospitalised ($X_{IH}$); bacterial pneumonia infected and COVID-19 recovered ($X_{IR}$); bacterial pneumonia recovered and COVID-19 susceptible ($X_{RS}$); bacterial pneumonia recovered and COVID-19 mildly infected ($X_{RI}$); bacterial pneumonia recovered and COVID-19 hospitalised ($X_{RH}$); and bacterial pneumonia recovered and COVID-19 recovered ($X_{RR}$). Additionally, we consider a compartment $B$ representing concentration of bacteria in the hospital environment. We make the following assumptions:
\begin{enumerate}
	\item COVID-19 is transmitted by contact with people in the $X_{SI}$, $X_{II}$ and $X_{RI}$ compartments.
	
	\item The population susceptible to COVID-19 are infected by this disease at a rate $\alpha_1$ if they have bacterial pneumonia, and at a rate $\alpha$ otherwise.
	
	\item The hospitalisation rate for people coinfected with COVID-19 and community-acquired pneumonia increases by an amount $\eta_1$ with respect to people with only COVID-19.
	
	\item The COVID-19 recovery rate for hospitalised people is $\theta_1$ if they are coinfected, and $\theta$ otherwise.
	
	\item Non-hospitalised people get community-acquired pneumonia by contact with people in the $X_{IS}$, $X_{II}$ and $X_{IR}$ compartments.
	
	\item Non-hospitalised people are infected with pneumonia at a rate $b_1$ if they have COVID-19, and at a rate $b$ otherwise.
	
	\item People hospitalised due to COVID-19 get hospital-acquired pneumonia at a rate proportional to the concentration of bacteria in the environment.
	
	\item The disease-induced death rate for coinfected hospitalised patients is increased by an amount $\delta_2$ with respect to those with only COVID-19.
	
	\item The pneumonia-induced death rate for non-hospitalised people is $\delta_0$ if they have COVID-19, and $\delta$ otherwise.
	
	\item The pneumonia recovery rate is $\phi_1$ for people in the $X_{II}$ compartment, $\phi_2$ for the $X_{IH}$ compartment, and $\phi$ for the $X_{IS}$ and $X_{IR}$ compartments.
\end{enumerate}

\begin{figure}
	\centering
	\includegraphics[width=\linewidth]{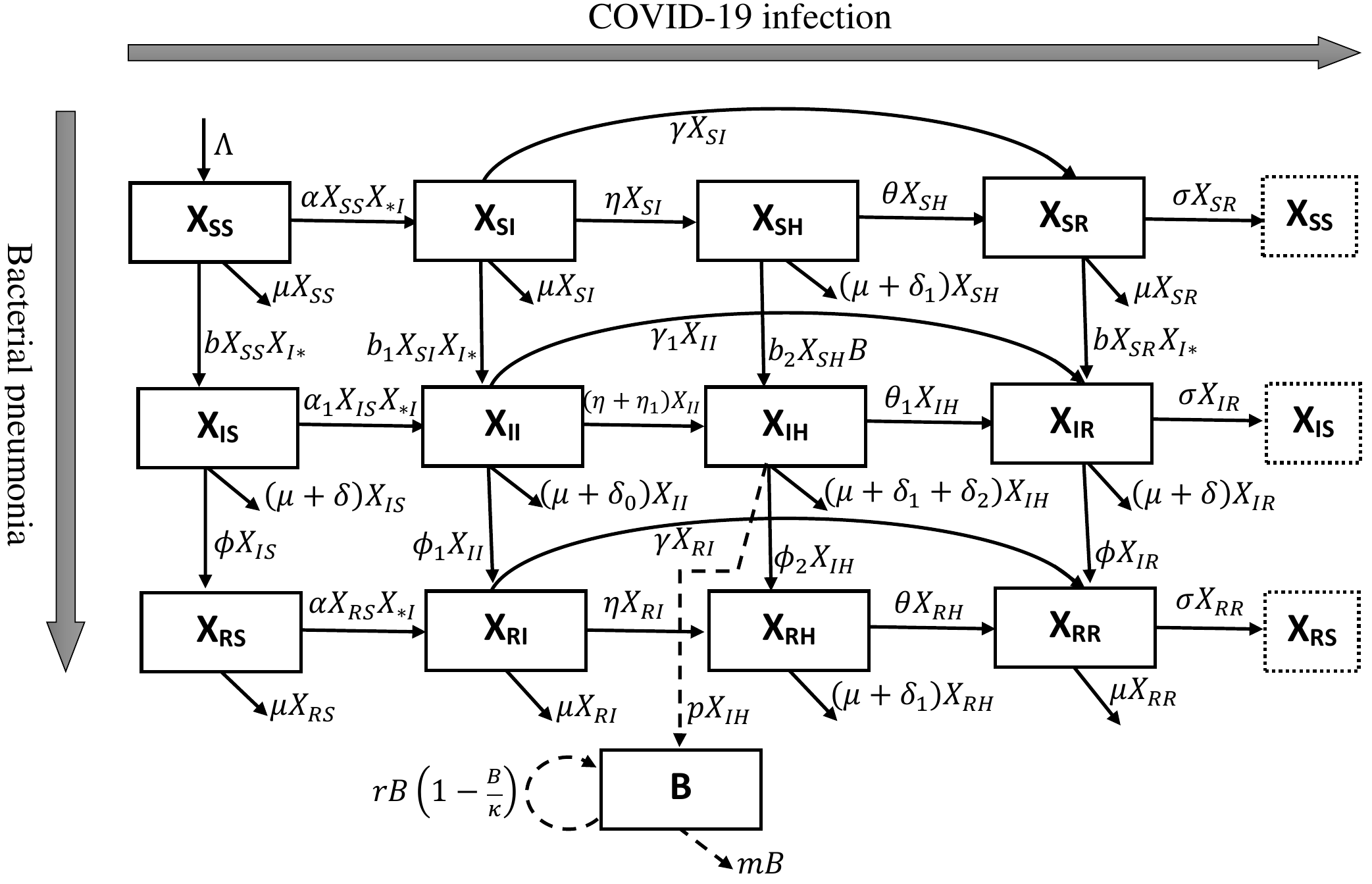}
	\caption{Schematic diagram of the coinfection model. Solid lines represent transition between compartments. Dashed lines represent proliferation of bacteria. $X_{*I}$ denotes $X_{SI}+X_{II}+X_{RI}$ and $X_{I*}$ denotes $X_{IS}+X_{II}+X_{IR}$.}
	\label{fig:drawing}
\end{figure}

The above assumptions yield a coinfection model given by the following system of 13 differential equations:

\begin{equation}\label{syst}
	\begin{aligned}
		X_{SS}' & = \Lambda + \sigma X_{SR} - \mu X_{SS} - \alpha X_{SS}\left(X_{SI}+X_{II}+X_{RI}\right) - bX_{SS}\left(X_{IS}+X_{II}+X_{IR}\right),                                     \\
		X_{SI}' & = \alpha X_{SS}\left(X_{SI}+X_{II}+X_{RI}\right) - (\gamma + \eta + \mu)X_{SI} - b_1X_{SI}\left(X_{IS}+X_{II}+X_{IR}\right),                                            \\
		X_{SH}' & = \eta X_{SI} - \theta X_{SH} - (\mu + \delta_1)X_{SH} - b_2X_{SH}B,                                                                                                    \\
		X_{SR}' & = \gamma X_{SI} + \theta X_{SH} - \mu X_{SR} - \sigma X_{SR} - b X_{SR}\left(X_{IS}+X_{II}+X_{IR}\right),                                                               \\
		X_{IS}' & = \sigma X_{IR} + bX_{SS}\left(X_{IS}+X_{II}+X_{IR}\right) - \alpha_1X_{IS}\left(X_{SI}+X_{II}+X_{RI}\right) - (\mu+\delta)X_{IS} - \phi X_{IS},                        \\
		X_{II}' & = b_1X_{SI}\left(X_{IS}+X_{II}+X_{IR}\right) + \alpha_1X_{IS}\left(X_{SI}+X_{II}+X_{RI}\right) - \left(\gamma_1 + \eta + \eta_1 + \mu + \delta_0 + \phi_1\right)X_{II}, \\
		X_{IH}' & = \left(\eta+\eta_1\right)X_{II} + b_2X_{SH}B - \theta_1X_{IH} - (\mu + \delta_1 + \delta_2)X_{IH} - \phi_2X_{IH},                                                      \\
		X_{IR}' & = b X_{SR}\left(X_{IS}+X_{II}+X_{IR}\right) + \gamma_1X_{II} + \theta_1X_{IH} - (\mu+\delta)X_{IR} - \phi X_{IR} - \sigma X_{IR},                                       \\
		X_{RS}' & = \sigma X_{RR} + \phi X_{IS} - \mu X_{RS} - \alpha X_{RS}\left(X_{SI}+X_{II}+X_{RI}\right),                                                                            \\
		X_{RI}' & = \phi_1 X_{II} + \alpha X_{RS}\left(X_{SI}+X_{II}+X_{RI}\right) - (\gamma + \eta + \mu)X_{RI},                                                                         \\
		X_{RH}' & = \eta X_{RI} + \phi_2 X_{IH} - \theta X_{RH} - (\mu + \delta_1)X_{RH},                                                                                                 \\
		X_{RR}' & = \phi X_{IR} + \gamma X_{RI} + \theta X_{RH} - \mu X_{RR} - \sigma X_{RR},                                                                                             \\
		B'      & = pX_{IH} + rB\left(1-\frac{B}{\kappa}\right) - mB.
	\end{aligned}
\end{equation}

The schematic diagram of model \eqref{syst} can be seen in Figure \ref{fig:drawing}. All parameters are assumed to be positive.

\section{Analysis of sub-models}
\label{sec:sub-models}

Before studying the dynamics of the coinfection model \eqref{syst}, we will analyse the two sub-models (COVID-19 only and bacterial pneumonia only).

\subsection{Analysis of the COVID-19 infection model}

The COVID-19-only model \eqref{COVID-model} has a disease-free equilibrium (DFE) given by
\[\mathcal{E}_{C0} = (S,I,H,R) = \left(\frac{\Lambda}{\mu},0,0,0\right).\]

The stability of $\mathcal{E}_{C0}$ depends on the basic reproduction number of model \eqref{COVID-model}. Using the notation in \cite{driessche2002reproduction}, we define the matrices $F$ and $V$ given by
\[F =
\begin{bmatrix}
	\frac{\alpha\Lambda}{\mu} & 0 \\
	0                         & 0
\end{bmatrix},\qquad
V =
\begin{bmatrix}
	\gamma+\eta+\mu & 0                 \\
	-\eta           & \theta+\delta+\mu
\end{bmatrix}.\]

Then, the basic reproduction number $\mathcal{R}_C$ of the COVID-19-only model is given by the spectral radius of $FV^{-1}$. From this, we obtain
\begin{equation}
\mathcal{R}_C = \frac{\alpha\Lambda}{\mu(\gamma+\eta+\mu)}.
\end{equation}

By \cite[Theorem 2]{driessche2002reproduction}, we obtain the following result.

\begin{theorem}
    The disease-free equilibrium $\mathcal{E}_{C0}$ of model \eqref{COVID-model} is locally asymptotically stable if $\mathcal{R}_C < 1$, but unstable if $\mathcal{R}_C > 1$.
\end{theorem}

\subsection{Analysis of the bacterial pneumonia infection model}

The bacterial pneumonia model \eqref{bacterial-model} has a DFE given by
\[\mathcal{E}_{P0} = (S,I,R,B) = \left(\frac{\Lambda}{\mu},0,0,0\right).\]

To apply the next-generation matrix method, we will compute the matrix of new infections $F$ and the transition matrix $V$, which are given by
\[F =
\begin{bmatrix}
	\frac{b\Lambda}{\mu}
\end{bmatrix},\qquad
V =
\begin{bmatrix}
	\phi+\mu+\delta
\end{bmatrix}.\]

Using the same method as before, we obtain the basic reproduction number $\mathcal{R}_P$ of the bacterial pneumonia-only model, which is
\begin{equation}
\mathcal{R}_P = \frac{b\Lambda}{\mu(\phi+\mu+\delta)}.
\end{equation}

Using \cite[Theorem 2]{driessche2002reproduction} again, we obtain the following result.

\begin{theorem}
    The disease-free equilibrium $\mathcal{E}_{P0}$ of model \eqref{bacterial-model} is locally asymptotically stable if $\mathcal{R}_P < 1$, but unstable if $\mathcal{R}_P > 1$.
\end{theorem}

\section{Analysis of the COVID-19--bacterial pneumonia coinfection model}
\label{sec:coinfection}

Next, we consider the dynamics of the coinfection model \eqref{syst}. The existence and stability of equilibria for model \eqref{syst} will depend on three parameters, which are defined as follows:
\[\mathcal{R}_C := \frac{\alpha\Lambda}{\mu(\gamma+\eta+\mu)},\qquad
\mathcal{R}_P := \frac{b\Lambda}{\mu(\phi+\mu+\delta)},\qquad
\mathcal{R}_B := \frac{r}{m}.\]
As we saw in the previous section, the parameters $\mathcal{R}_C$ and $\mathcal{R}_P$ represent the basic reproduction numbers of COVID-19 and bacterial pneumonia, respectively. On the other hand, $\mathcal{R}_B$ can be interpreted as the reproduction number of bacterial population in the hospital.

\subsection{Equilibria of the model}

By direct computation, we obtain the following result about the equilibria of model \eqref{syst}.

\begin{theorem}\label{teo:equilibria}
    The coinfection model \eqref{syst} has the following steady states:
    
    \begin{enumerate}
        \item The disease-free, bacterial population-free equilibrium:
        \begin{equation*}
        \mathcal{E}_0 = \left(X_{SS}^{(0)},0,0,0,0,0,0,0,0,0,0,0,0\right),
        \end{equation*}
        where
        \[X_{SS}^{(0)} = \frac{\Lambda}{\mu}.\]
        
        \item The disease-free, bacterial population-present equilibrium:
        \begin{equation*}
        \mathcal{E}_1 = \left(X_{SS}^{(1)},0,0,0,0,0,0,0,0,0,0,0,B^{(1)}\right),
        \end{equation*}
        where
        \[X_{SS}^{(1)} = \frac{\Lambda}{\mu},\qquad
        B^{(1)} = \frac{\kappa}{r}(r-m).\]
        This equilibrium exists if and only if $\mathcal{R}_B>1$.
        
        \item The COVID-19-free, pneumonia-present, bacterial population-free equilibrium:
        \begin{equation*}
        \mathcal{E}_2 = \left(X_{SS}^{(2)},0,0,0,X_{IS}^{(2)},0,0,0,X_{RS}^{(2)},0,0,0,0\right),
        \end{equation*}
        where
        \[X_{SS}^{(2)} = \frac{\mu+\delta+\phi}{b},\qquad
        X_{IS}^{(2)} = \frac{\Lambda}{\mu+\delta+\phi}-\frac{\mu}{b},\qquad
        X_{RS}^{(2)} = \frac{\phi}{\mu}X_{IS}^{(2)}.\]
        This equilibrium exists if and only if $\mathcal{R}_P > 1$.
        
        \item The COVID-19-free, pneumonia-present, bacterial population-present equilibrium:
        \begin{equation*}
        \mathcal{E}_3 = \left(X_{SS}^{(3)},0,0,0,X_{IS}^{(3)},0,0,0,X_{RS}^{(3)},0,0,0,B^{(3)}\right),
        \end{equation*}
        where
        \begin{align*}
        	 & X_{SS}^{(3)} = \frac{\mu+\delta+\phi}{b},\qquad
        X_{IS}^{(3)} = \frac{\Lambda}{\mu+\delta+\phi}-\frac{\mu}{b},\qquad
        X_{RS}^{(3)} = \frac{\phi}{\mu}X_{IS}^{(3)}, \\
        	 & B^{(3)} = \frac{\kappa}{r}(r-m).
        \end{align*}
        This equilibrium exists if and only if
        \[\mathcal{R}_B>1
        \quad\text{and}\quad
        \mathcal{R}_P > 1.\]
        
        \item The COVID-19-present, pneumonia-free, bacterial population-free equilibrium:
        \begin{equation*}
        \mathcal{E}_4 = \left(X_{SS}^{(4)},X_{SI}^{(4)},X_{SH}^{(4)},X_{SR}^{(4)},0,0,0,0,0,0,0,0,0\right),
        \end{equation*}
        where
        \begin{align*}
        	 & X_{SS}^{(4)} = \frac{\gamma+\eta+\mu}{\alpha},\quad
        X_{SI}^{(4)} = \frac{(\mu+\sigma)(\theta+\mu+\delta_1)\big[\alpha\Lambda-\mu(\gamma+\eta+\mu)\big]}{\alpha\big[\mu(\theta+\mu+\delta_1)(\gamma+\eta+\mu+\sigma) + \eta\sigma(\mu+\delta_1)\big]}, \\
        	 & X_{SH}^{(4)} = \frac{\eta}{\theta+\mu+\delta_1}X_{SI}^{(4)},\quad
        X_{SR}^{(4)} = \left(\frac{\gamma}{\mu+\sigma} + \frac{\eta\theta}{(\mu+\sigma)(\theta+\mu+\delta_1)}\right)X_{SI}^{(4)}.
        \end{align*}
        This equilibrium exists if and only if
        \[\mathcal{R}_C > 1.\]
    \end{enumerate}
\end{theorem}
\begin{proof}
	Equilibria $\mathcal{E}_0$, $\mathcal{E}_1$, $\mathcal{E}_2$ and $\mathcal{E}_3$ are obtained by assuming that $X_{SI}=0$ in the system at equilibrium and solving the resulting algebraic equations. This yields four different cases: one for each equilibrium.
	
	On the other hand, assuming $X_{SI}>0$ and $X_{IS}=0$ results in only one case, corresponding to the equilibrium $\mathcal{E}_4$.
	
	The case when $X_{SI}>0$ and $X_{IS}>0$ will be discussed below.
\end{proof}

Theorem \ref{teo:equilibria} shows that, under certain conditions, the coinfection model has five different steady states. Moreover, we conjecture that a sixth equilibrium, with positive values for all variables, may exist. We will denote this interior equilibrium by $\mathcal{E}_5$. Since the theoretical analysis becomes too cumbersome in this case, we will resort to numerical simulations to investigate the dynamics of equilibrium $\mathcal{E}_5$ (see Section \ref{sec:num}).

\subsection{Stability analysis}

We will now analyse the local stability for the equilibria of system \eqref{syst}. Our results focus only on the disease-free equilibria $\mathcal{E}_0$ and $\mathcal{E}_1$.

\begin{theorem}\ 
	\begin{enumerate}
		\item[(i)] The disease-free, bacterial population-free equilibrium $\mathcal{E}_0$ is locally asymptotically stable if and only if
		\begin{equation}\label{stable-E0}
			\mathcal{R}_C < 1,\quad
			\mathcal{R}_P < 1\quad\text{and}\quad
			\mathcal{R}_B < 1.
		\end{equation}
		
		\item[(ii)] The disease-free, bacterial population-present equilibrium $\mathcal{E}_1$ is locally asymptotically stable if and only if
		\begin{equation}\label{stable-E1}
			\mathcal{R}_C < 1,\quad
			\mathcal{R}_P < 1\quad\text{and}\quad
			\mathcal{R}_B > 1.
		\end{equation}
	\end{enumerate}
\end{theorem}
\begin{proof}
	The Jacobian of system \eqref{syst} evaluated at $\mathcal{E}_0$ is given by
	{\small 
	\begin{equation*}
		J_0 =
		\begin{bmatrix}
			-\mu & -\aLm    & 0      & \sigma & -\bLm    & -\tfrac{(\alpha+b)\Lambda}{\mu} & 0             & -\bLm       & 0    & -\aLm  & 0      & 0      & 0   \\
			0    & \aLm-k_1 & 0      & 0      & 0        & \aLm                            & 0             & 0           & 0    & \aLm   & 0      & 0      & 0   \\
			0    & \eta     & -k_2   & 0      & 0        & 0                               & 0             & 0           & 0    & 0      & 0      & 0      & 0   \\
			0    & \gamma   & \theta & -k_3   & 0        & 0                               & 0             & 0           & 0    & 0      & 0      & 0      & 0   \\
			0    & 0        & 0      & 0      & \bLm-k_4 & \bLm                            & 0             & \bLm+\sigma & 0    & 0      & 0      & 0      & 0   \\
			0    & 0        & 0      & 0      & 0        & -k_5                            & 0             & 0           & 0    & 0      & 0      & 0      & 0   \\
			0    & 0        & 0      & 0      & 0        & \eta+\eta_1                     & -\theta_1-k_6 & 0           & 0    & 0      & 0      & 0      & 0   \\
			0    & 0        & 0      & 0      & 0        & \gamma_1                        & \theta_1      & -k_7        & 0    & 0      & 0      & 0      & 0   \\
			0    & 0        & 0      & 0      & \phi     & 0                               & 0             & 0           & -\mu & 0      & 0      & \sigma & 0   \\
			0    & 0        & 0      & 0      & 0        & \phi_1                          & 0             & 0           & 0    & -k_1   & 0      & 0      & 0   \\
			0    & 0        & 0      & 0      & 0        & 0                               & \phi_2        & 0           & 0    & \eta   & -k_2   & 0      & 0   \\
			0    & 0        & 0      & 0      & 0        & 0                               & 0             & \phi        & 0    & \gamma & \theta & -k_3   & 0   \\
			0    & 0        & 0      & 0      & 0        & 0                               & p             & 0           & 0    & 0      & 0      & 0      & r-m
		\end{bmatrix}
	\end{equation*}}
	where
	\begin{align*}
		 & k_1 = \gamma+\eta+\mu,\qquad
		k_2 = \theta+\mu+\delta_1,\qquad
		k_3 = \mu+\sigma,\qquad
		k_4 = \mu+\delta+\phi,                \\
		 & k_5 = \gamma_1+\eta+\eta_1+\mu+\delta_0+\phi_1,\qquad
		k_6 = \mu+\delta_1+\delta_2+\phi_2,\qquad
		k_7 = \mu+\delta+\phi+\sigma.
	\end{align*}
	
	From this, we obtain the characteristic polynomial
	\begin{align*}
		 & \left(\lambda+\mu\right)^2 \left(\lambda+k_1\right) \left(\lambda+k_2\right)^2 \left(\lambda+k_3\right)^2 \left(\lambda+k_5\right) \left(\lambda+\theta_1+k_6\right) \left(\lambda+k_7\right) \\
		 & \quad\times \left(\lambda+k_1-\frac{\alpha\Lambda}{\mu}\right) \left(\lambda+k_4-\frac{b\Lambda}{\mu}\right) \left(\lambda+m-r\right)
		= 0.
	\end{align*}
	
	Due to positivity of parameters, it is clear that all eigenvalues have negative real part if and only if
	\[\gamma+\eta+\mu-\frac{\alpha\Lambda}{\mu} > 0,\qquad
	\mu+\delta+\phi-\frac{b\Lambda}{\mu} > 0\quad
	\text{and}\qquad
	m-r > 0,\]
	which is equivalent to the condition \eqref{stable-E0}. This proves part (i) of the theorem.
	
	Next, we compute the Jacobian at $\mathcal{E}_1$, which is given by
	{\small 
	\begin{equation*}
		J_1 =
		\begin{bmatrix}
			-\mu & -\aLm    & 0        & \sigma & -\bLm    & -\tfrac{(\alpha+b)\Lambda}{\mu} & 0             & -\bLm       & 0    & -\aLm  & 0      & 0      & 0   \\
			0    & \aLm-k_1 & 0        & 0      & 0        & \aLm                            & 0             & 0           & 0    & \aLm   & 0      & 0      & 0   \\
			0    & \eta     & -k_0-k_2 & 0      & 0        & 0                               & 0             & 0           & 0    & 0      & 0      & 0      & 0   \\
			0    & \gamma   & \theta   & -k_3   & 0        & 0                               & 0             & 0           & 0    & 0      & 0      & 0      & 0   \\
			0    & 0        & 0        & 0      & \bLm-k_4 & \bLm                            & 0             & \bLm+\sigma & 0    & 0      & 0      & 0      & 0   \\
			0    & 0        & 0        & 0      & 0        & -k_5                            & 0             & 0           & 0    & 0      & 0      & 0      & 0   \\
			0    & 0        & k_0      & 0      & 0        & \eta+\eta_1                     & -\theta_1-k_6 & 0           & 0    & 0      & 0      & 0      & 0   \\
			0    & 0        & 0        & 0      & 0        & \gamma_1                        & \theta_1      & -k_7        & 0    & 0      & 0      & 0      & 0   \\
			0    & 0        & 0        & 0      & \phi     & 0                               & 0             & 0           & -\mu & 0      & 0      & \sigma & 0   \\
			0    & 0        & 0        & 0      & 0        & \phi_1                          & 0             & 0           & 0    & -k_1   & 0      & 0      & 0   \\
			0    & 0        & 0        & 0      & 0        & 0                               & \phi_2        & 0           & 0    & \eta   & -k_2   & 0      & 0   \\
			0    & 0        & 0        & 0      & 0        & 0                               & 0             & \phi        & 0    & \gamma & \theta & -k_3   & 0   \\
			0    & 0        & 0        & 0      & 0        & 0                               & p             & 0           & 0    & 0      & 0      & 0      & m-r
		\end{bmatrix},
	\end{equation*}}
	where $k_0 = b_2\kappa\left(1-\frac m r\right)$, and $k_1,\ldots,k_7$ are as defined above. Notice that $k_0>0$ if and only if $\mathcal{R}_B>1$.
	
	The characteristic polynomial at $\mathcal{E}_1$ is
	\begin{align*}
		 & \left(\lambda+\mu\right)^2 \left(\lambda+k_1\right) \left(\lambda+k_2\right) \left(\lambda+k_3\right)^2 \left(\lambda+k_5\right) \left(\lambda+\theta_1+k_6\right) \left(\lambda+k_7\right) \\
		 & \quad\times
		 \left(\lambda+k_0+k_2\right) \left(\lambda+k_1-\frac{\alpha\Lambda}{\mu}\right) \left(\lambda+k_4-\frac{b\Lambda}{\mu}\right) \left(\lambda+r-m\right)
		= 0.
	\end{align*}
	Then, all eigenvalues have negative real part if and only if
	\[k_0+k_2 > 0,\quad
	\gamma+\eta+\mu-\frac{\alpha\Lambda}{\mu} > 0,\quad
	\mu+\delta+\phi-\frac{b\Lambda}{\mu} > 0\quad
	\text{and}\quad
	r-m > 0.\]
	The first of these inequalities holds automatically when $\mathcal{R}_B>1$. Hence, we can conclude that $\mathcal{E}_1$ is locally asymptotically stable if and only if the last three inequalities hold, which is equivalent to condition \eqref{stable-E1}. Thus, the proof of (ii) is complete.
\end{proof}

\section{Numerical analysis}
\label{sec:num}

In this section, we perform some simulations for system \eqref{syst} to illustrate the dynamics of the coinfection model in some cases that are not covered by the analysis in Section \ref{sec:coinfection}. We will consider the initial conditions
\begin{align*}
	 & X_{SS}(0) = 8.33\times10^7,\quad
	X_{SI}(0) = 10^5,\quad
	X_{SH}(0) = 10^3,\quad
	X_{SR}(0) = 10^5,\quad
	X_{IS}(0) = 10^3, \\
	 & B(0) = 0.8,\quad
	 X_{II}(0) = X_{IH}(0) = X_{IR}(0) = X_{RS}(0) = X_{RI}(0) = X_{RH}(0) = X_{RR}(0) = 0.
\end{align*}

Throughout this section, we will use the parameter values shown in Table \ref{tab:param}. These are based on the values used in other models for COVID-19, although they do not necessarily correspond to the dynamics in any specific country. Thus, we obtain a fixed value for $\mathcal{R}_C$, which is greater than one ($\mathcal{R}_C=1.2294$), while $\mathcal{R}_P$ and $\mathcal{R}_B$ will vary as the parameters $b$ and $r$ take different values.

\begin{table}
	\caption{Parameter values used for the coinfection model.}
	\label{tab:param}
	\centering
	\begin{tabular}{ccc}
		\hline
		Parameter  &       Value        &                 Unit                  \\ \hline
		$\Lambda$  &        2000        &              people/day               \\
		  $\mu$    & $2.4\times10^{-5}$ & $(\text{people}\cdot\text{day})^{-1}$ \\
		 $\sigma$  &       1/100        &           $\text{day}^{-1}$           \\
		 $\gamma$  &        1/12        &           $\text{day}^{-1}$           \\
		$\gamma_1$ &        1/20        &           $\text{day}^{-1}$           \\
		 $\theta$  &        1/14        &           $\text{day}^{-1}$           \\
		$\theta_1$ &        1/24        &           $\text{day}^{-1}$           \\
		  $b_1$    &  $2\times10^{-9}$  & $(\text{people}\cdot\text{day})^{-1}$ \\
		  $b_2$    &        0.1         &           $\text{day}^{-1}$           \\
		 $\delta$  &       0.001        &           $\text{day}^{-1}$           \\
		$\delta_0$ &       0.005        &           $\text{day}^{-1}$           \\
		$\delta_1$ &        0.01        &           $\text{day}^{-1}$           \\
		$\delta_2$ &        0.2         &           $\text{day}^{-1}$           \\
		  $\eta$   &        0.12        &           $\text{day}^{-1}$           \\
		 $\eta_1$  &        0.1         &           $\text{day}^{-1}$           \\
		  $\phi$   &        1/14        &           $\text{day}^{-1}$           \\
		 $\phi_1$  &        1/30        &           $\text{day}^{-1}$           \\
		 $\phi_2$  &        1/40        &           $\text{day}^{-1}$           \\
		   $p$     &     $10^{-5}$      & $(\text{people}\cdot\text{day})^{-1}$ \\
		 $\kappa$  &         1          &                                       \\
		   $m$     &        0.01        &           $\text{day}^{-1}$           \\
		 $\alpha$  &  $3\times10^{-9}$  & $(\text{people}\cdot\text{day})^{-1}$ \\
		$\alpha_1$ &     $10^{-8}$      & $(\text{people}\cdot\text{day})^{-1}$ \\
		   $b$     &      variable      & $(\text{people}\cdot\text{day})^{-1}$ \\
		   $r$     &      variable      &           $\text{day}^{-1}$           \\ \hline
	\end{tabular}
\end{table}

\textbf{Case 1.} When $b=10^{-10}$ and $r=0.004$, we have $\mathcal{R}_P = 0.1150 < 1$ and $\mathcal{R}_B = 0.4 < 1$. The time plots of the solutions for this case are shown in Figure \ref{fig:Rc}. The solutions converge to a positive equilibrium
\begin{align*}
	\mathcal{E}_5 & \approx \big(6.3418\times10^7,\ 5684,\ 3153,\ 69716,\ 191.8,\ 0.0735,\ 1537,\ 776.8, \\
	              & \qquad 4.368\times10^6,\ 391.5,\ 1048,\ 16263,\ 1.3487\big).
\end{align*}

\begin{figure}
	\begin{subfigure}[b]{0.5\linewidth}
		\centering
		\includegraphics[width=0.95\linewidth]{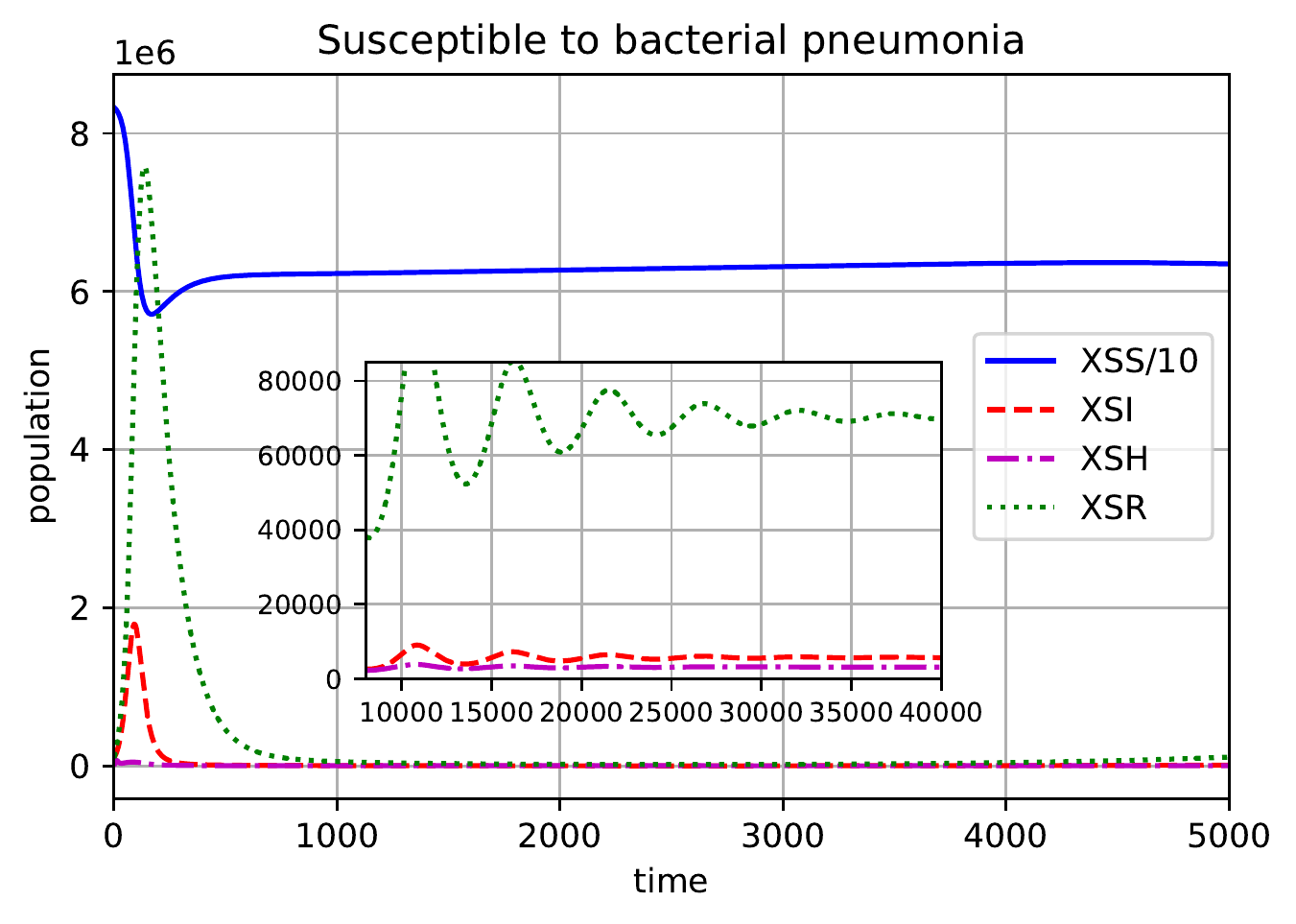}
		\vspace{4ex}
	\end{subfigure}%% 
	\begin{subfigure}[b]{0.5\linewidth}
		\centering
		\includegraphics[width=0.95\linewidth]{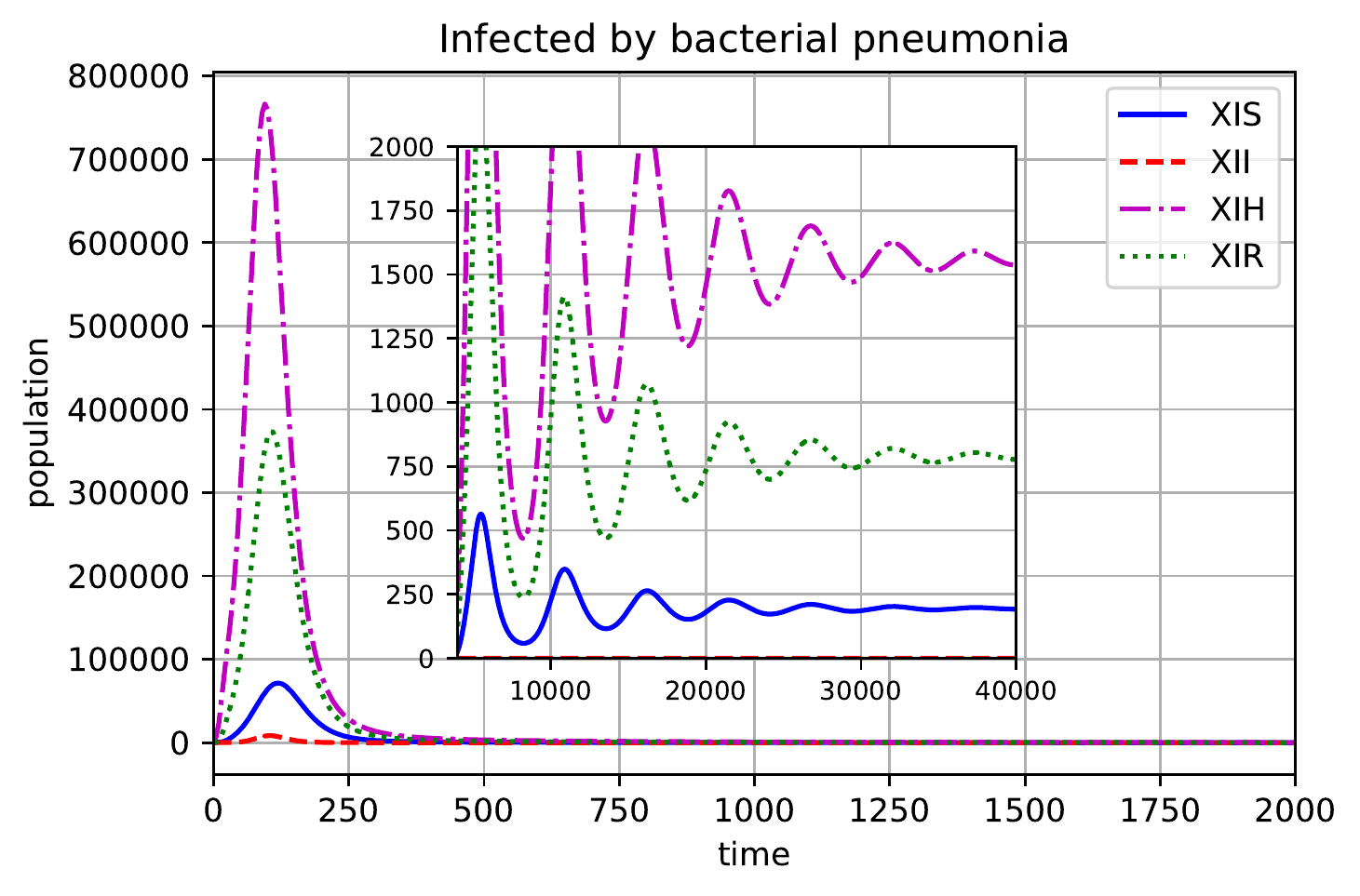}
		\vspace{4ex}
	\end{subfigure} 
	\begin{subfigure}[b]{0.5\linewidth}
		\centering
		\includegraphics[width=0.95\linewidth]{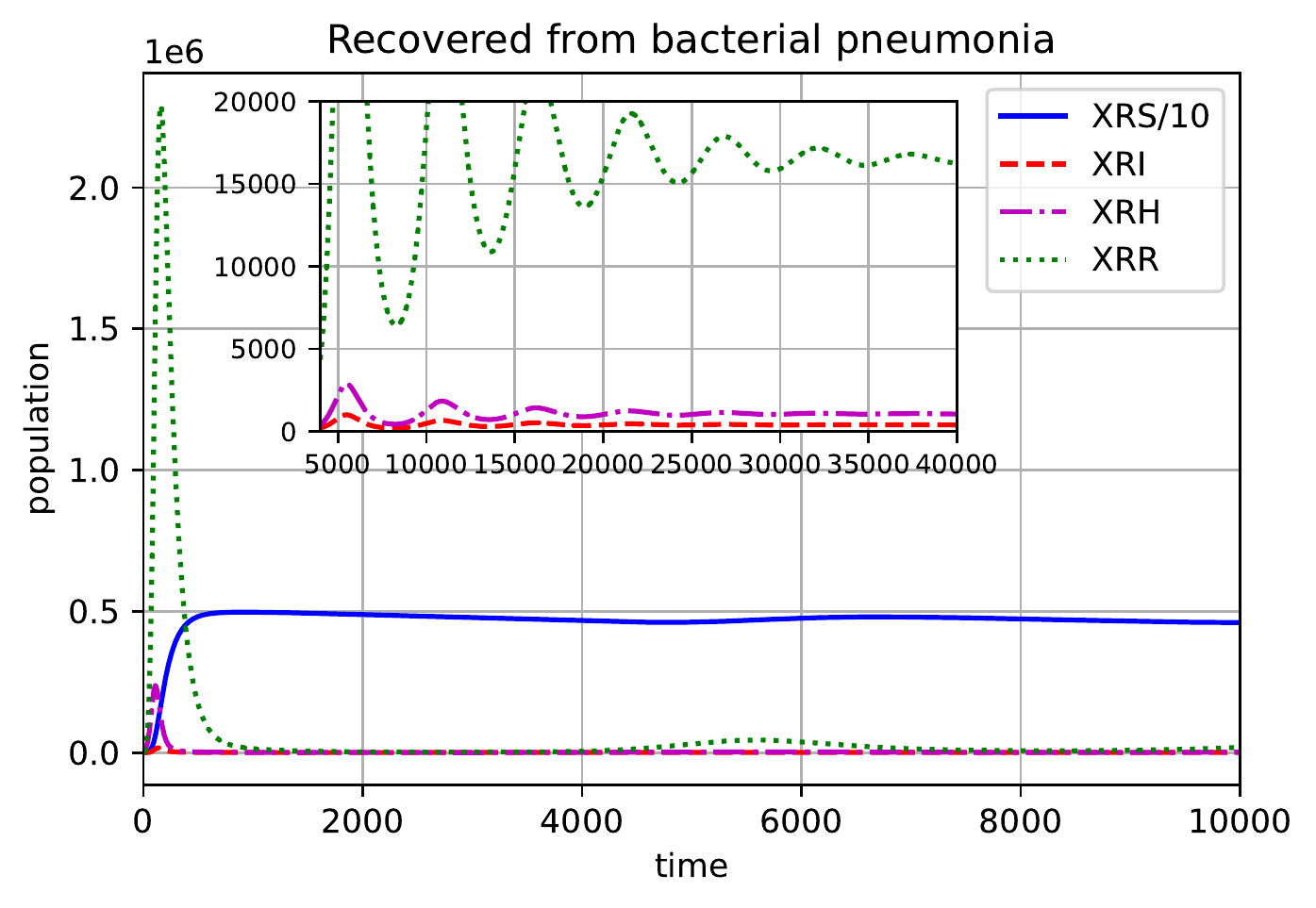}
	\end{subfigure}%%
	\begin{subfigure}[b]{0.5\linewidth}
		\centering
		\includegraphics[width=0.95\linewidth]{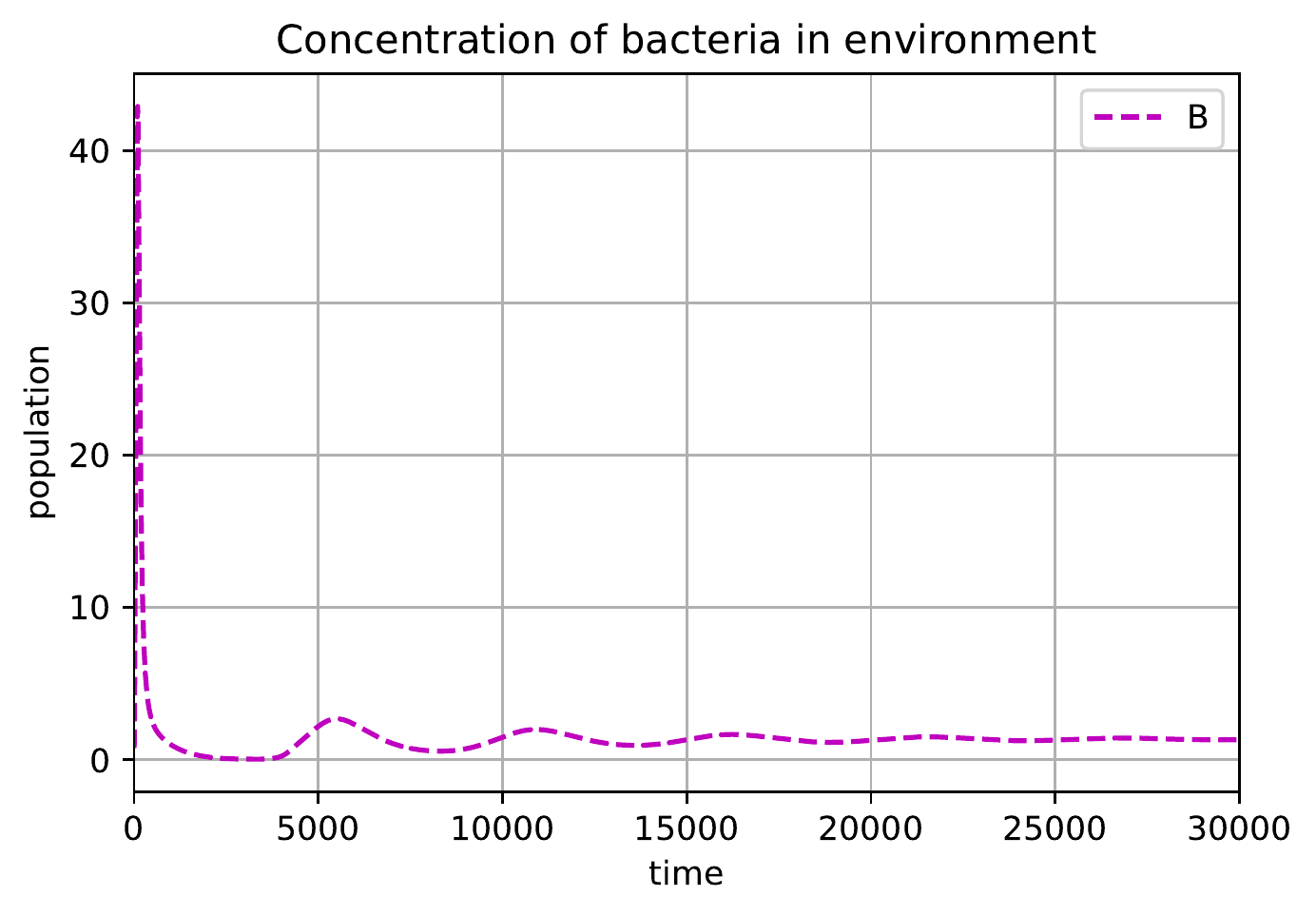}
	\end{subfigure} 
	\caption{Dynamics of the coinfection model when $\mathcal{R}_C>1$, $\mathcal{R}_P<1$ and $\mathcal{R}_B<1$.}
	\label{fig:Rc} 
\end{figure}

\textbf{Case 2.} When $b=9\times10^{-10}$ and $r=0.004$, we have $\mathcal{R}_P = 1.0352 > 1$ and $\mathcal{R}_B = 0.4 < 1$. The time plots of the solutions are depicted in Figure \ref{fig:RcRp}; we can see that they converge to a positive equilibrium
\begin{align*}
	\mathcal{E}_5 & \approx \big(5.711\times10^7,\ 5602,\ 3130,\ 6.89\times10^4,\ 2229,\ 0.59,\ 1509,\ 765, \\
	              & \qquad 1.066\times10^7,\ 1046,\ 2005,\ 2.85\times10^4,\ 1.33\big).
\end{align*}

\begin{figure}
	\begin{subfigure}[b]{0.5\linewidth}
		\centering
		\includegraphics[width=0.95\linewidth]{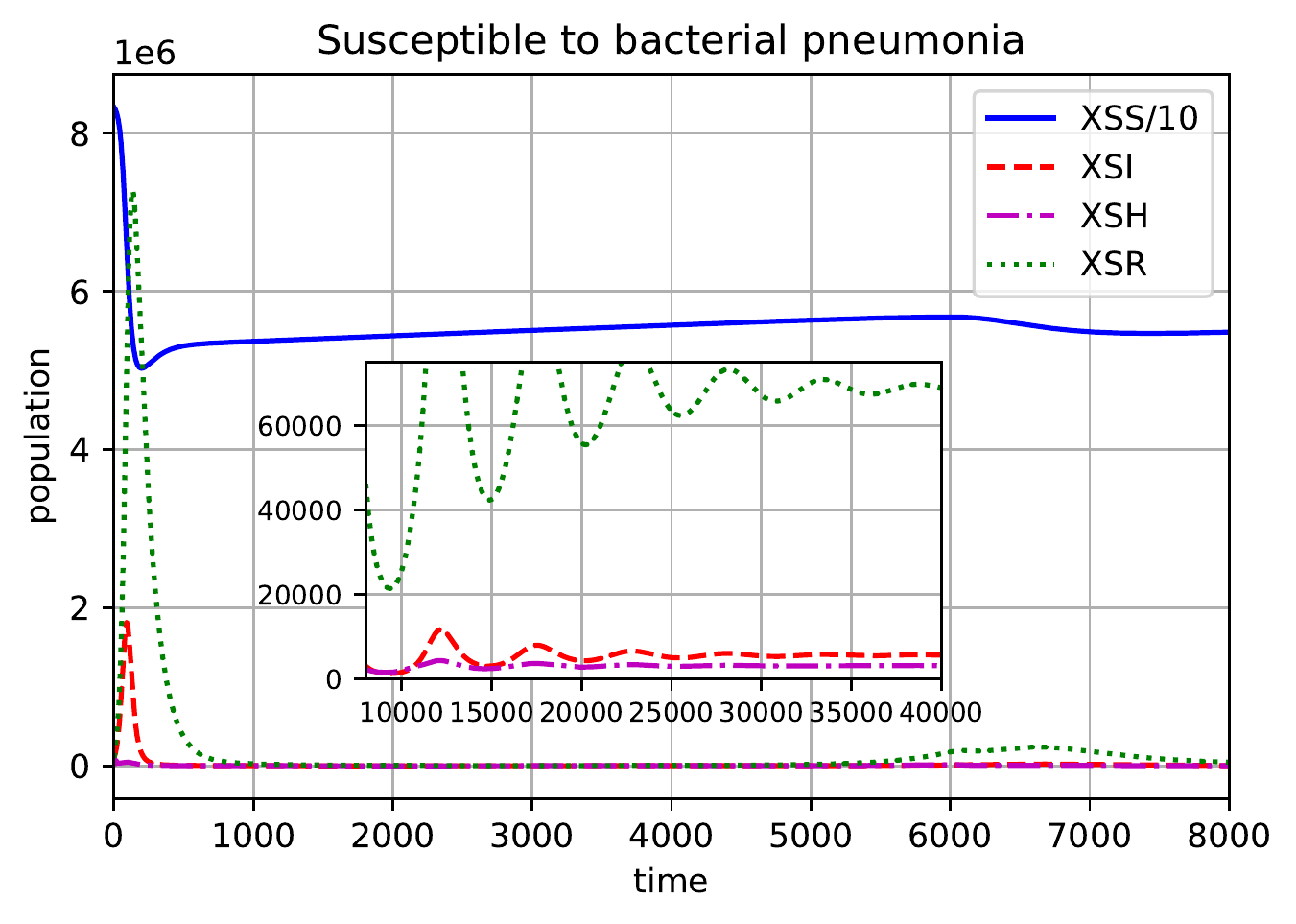}
		\vspace{4ex}
	\end{subfigure}%% 
	\begin{subfigure}[b]{0.5\linewidth}
		\centering
		\includegraphics[width=0.95\linewidth]{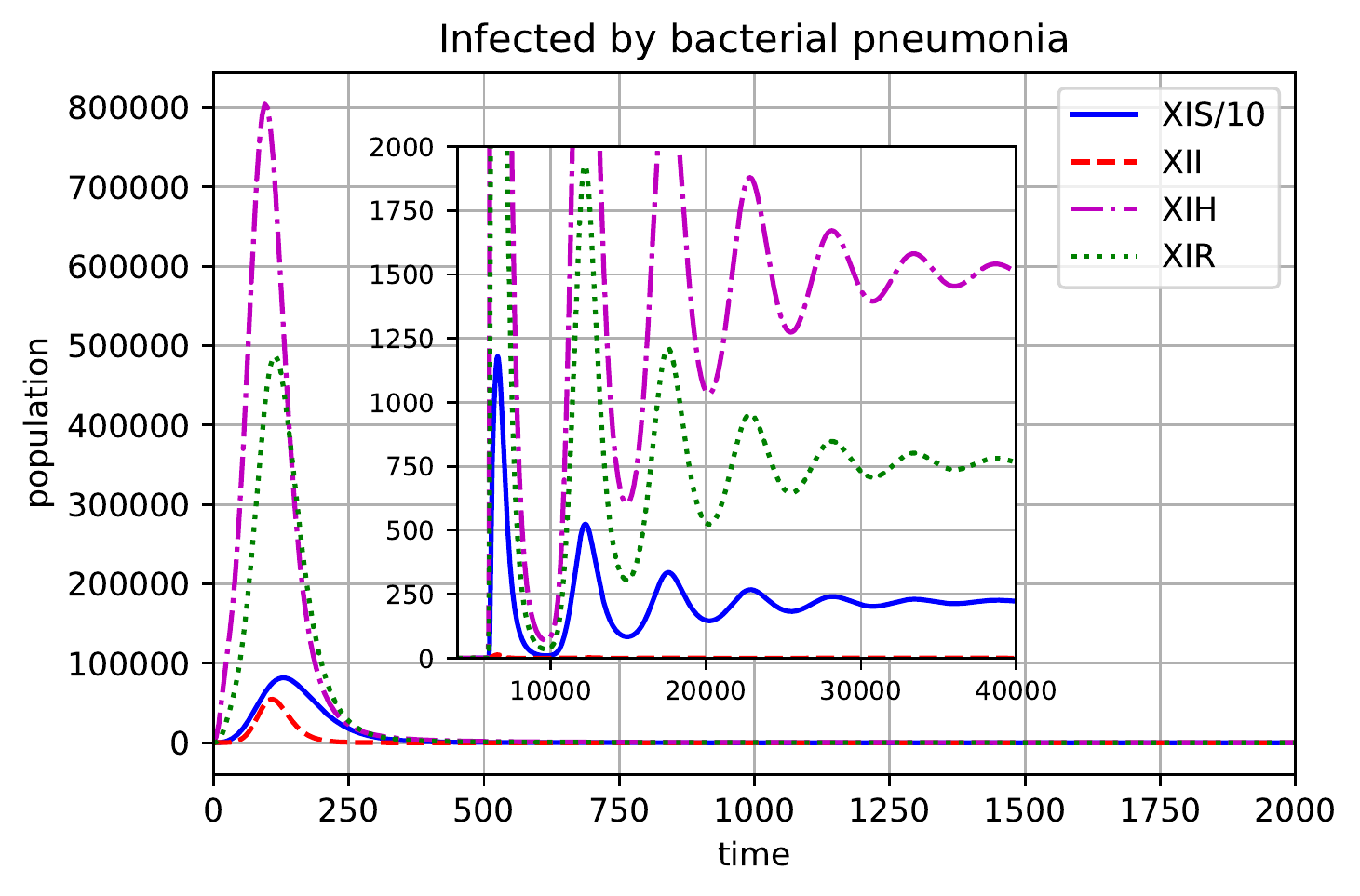}
		\vspace{4ex}
	\end{subfigure} 
	\begin{subfigure}[b]{0.5\linewidth}
		\centering
		\includegraphics[width=0.95\linewidth]{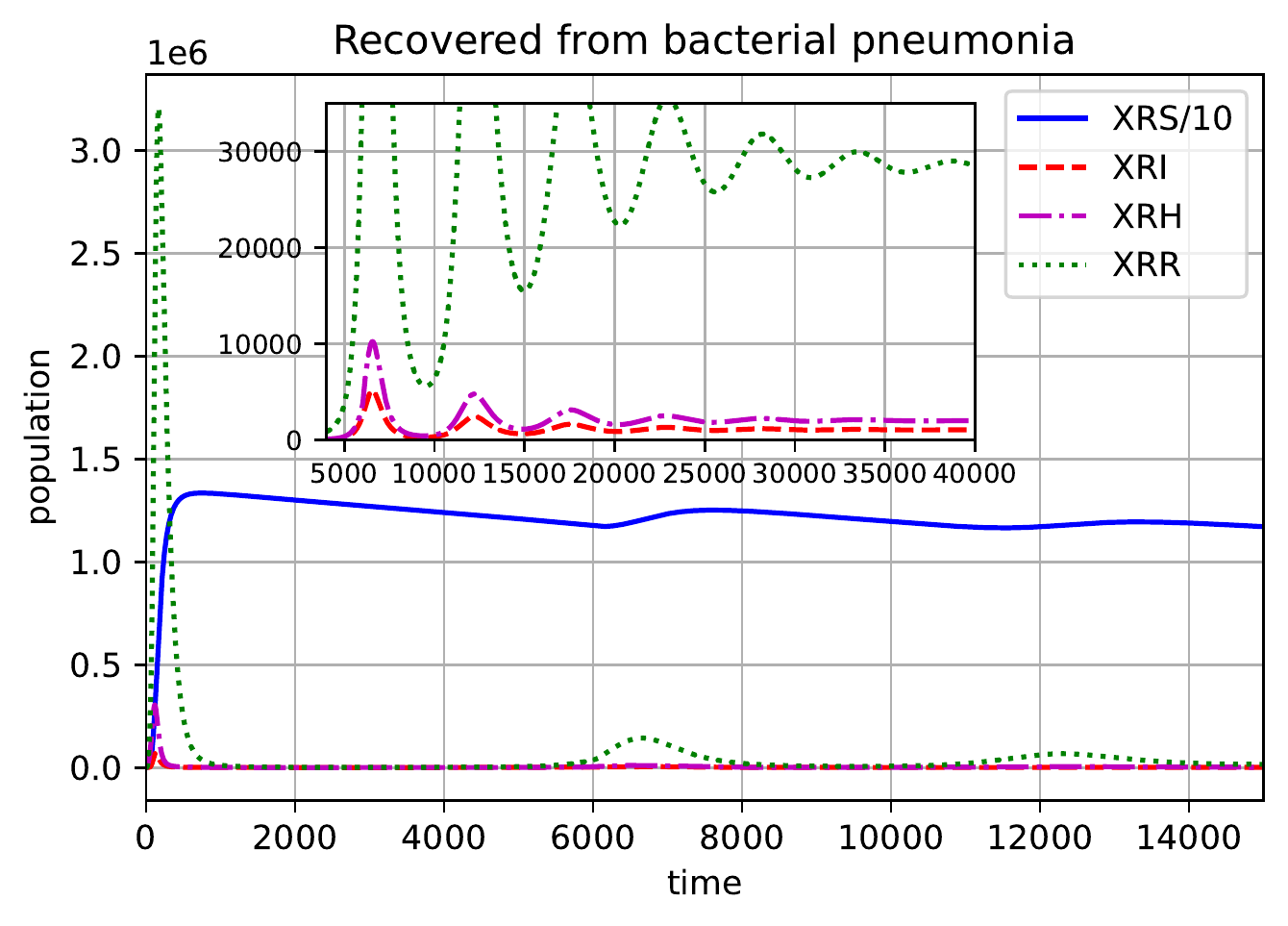}
	\end{subfigure}%%
	\begin{subfigure}[b]{0.5\linewidth}
		\centering
		\includegraphics[width=0.95\linewidth]{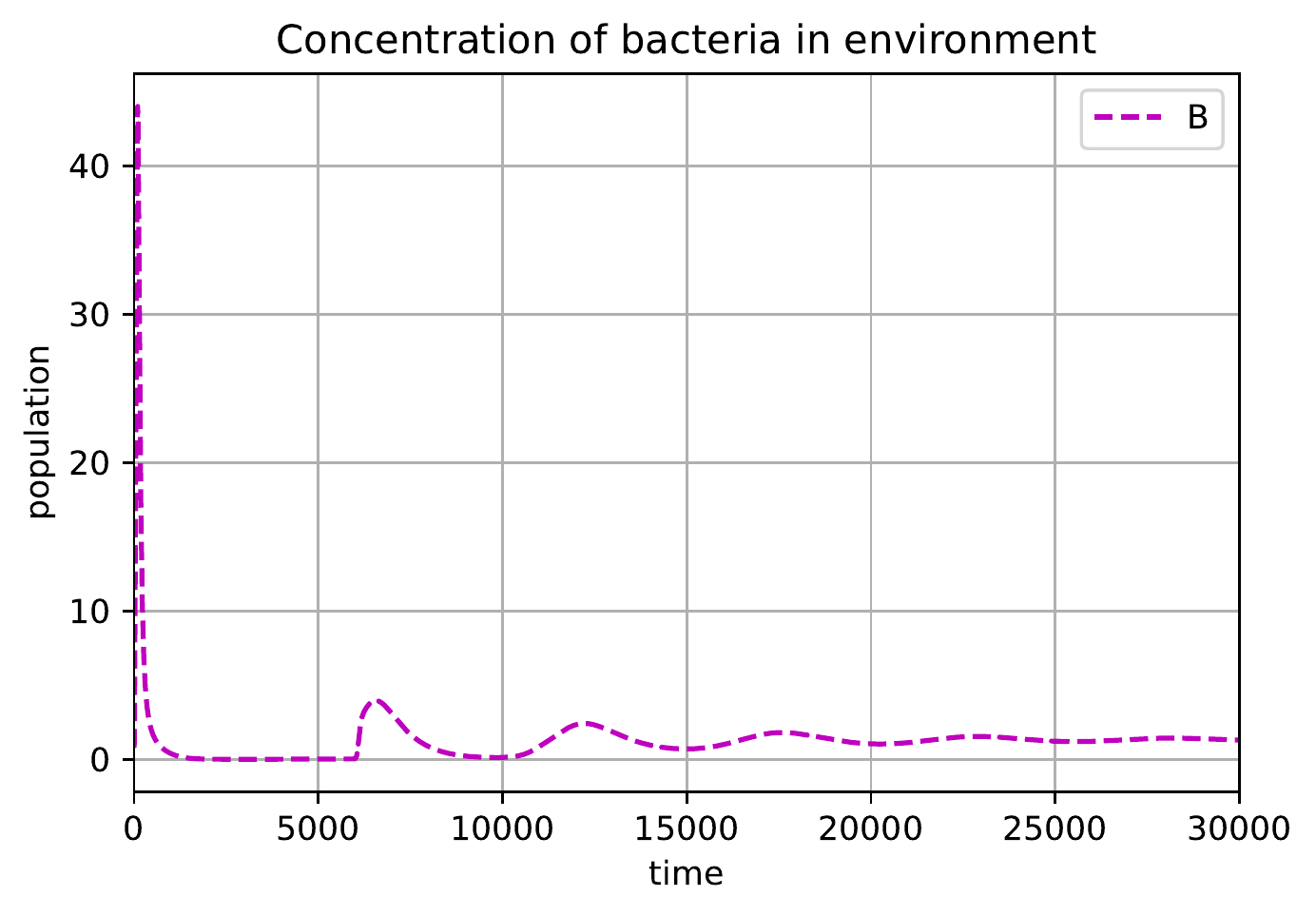}
	\end{subfigure} 
	\caption{Dynamics of the coinfection model when $\mathcal{R}_C>1$, $\mathcal{R}_P>1$ and $\mathcal{R}_B<1$.}
	\label{fig:RcRp} 
\end{figure}

\textbf{Case 3.} When $b=10^{-10}$ and $r=0.08$, we have $\mathcal{R}_P = 0.1150 < 1$ and $\mathcal{R}_B = 8 > 1$. The time plots of the solutions are depicted in Figure \ref{fig:RcRb}. We can see that they converge to the positive equilibrium
\begin{align*}
	\mathcal{E}_5 & \approx \big(6.353\times10^7,\ 6206,\ 3984,\ 8.0\times10^4,\ 189.8,\ 0.0793,\ 1519,\ 768, \\
	              & \qquad 4.251\times10^6,\ 415.3,\ 1078,\ 16605,\ 1.055\big).
\end{align*}

\begin{figure}
	\begin{subfigure}[b]{0.5\linewidth}
		\centering
		\includegraphics[width=0.95\linewidth]{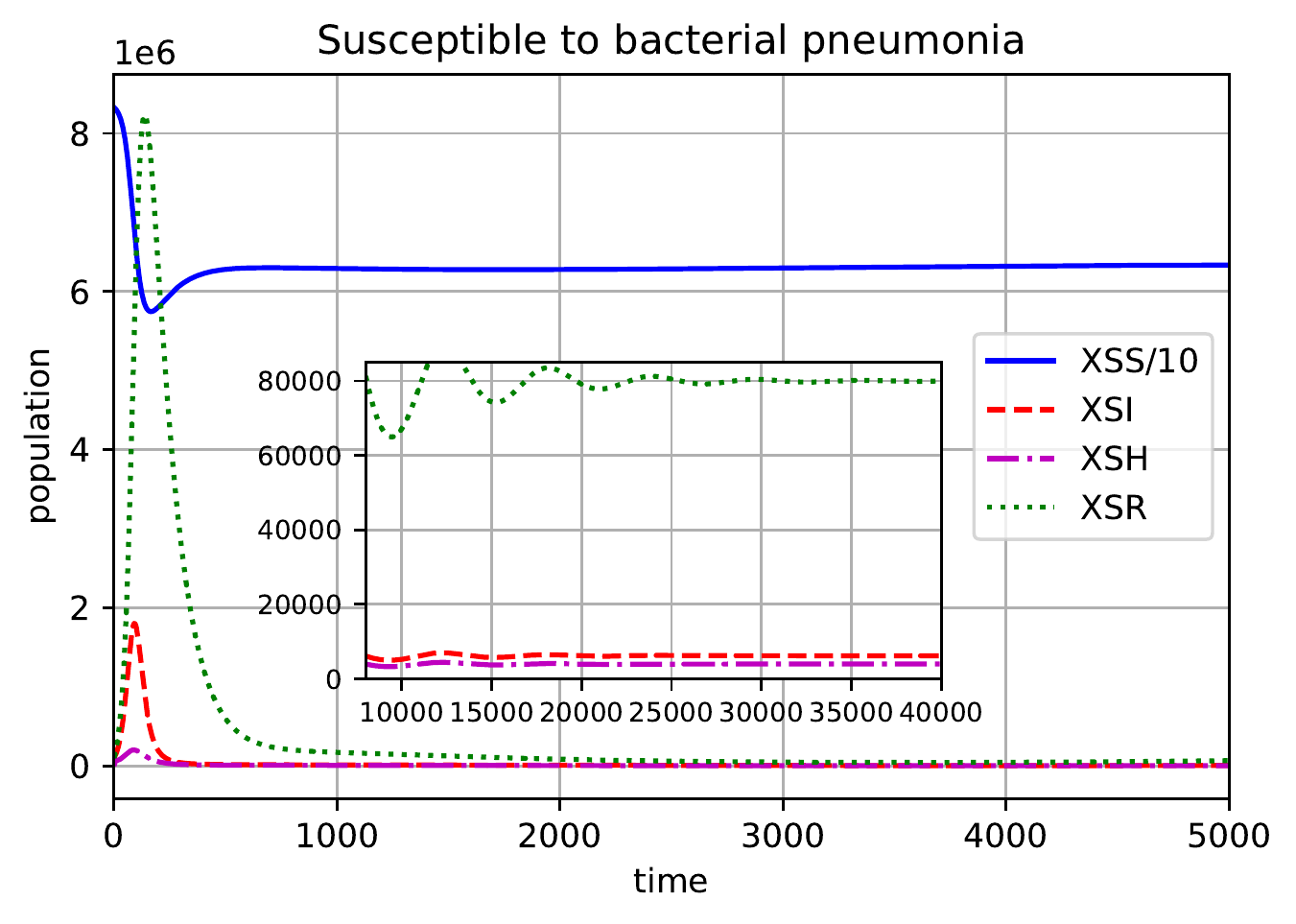}
		\vspace{4ex}
	\end{subfigure}%% 
	\begin{subfigure}[b]{0.5\linewidth}
		\centering
		\includegraphics[width=0.95\linewidth]{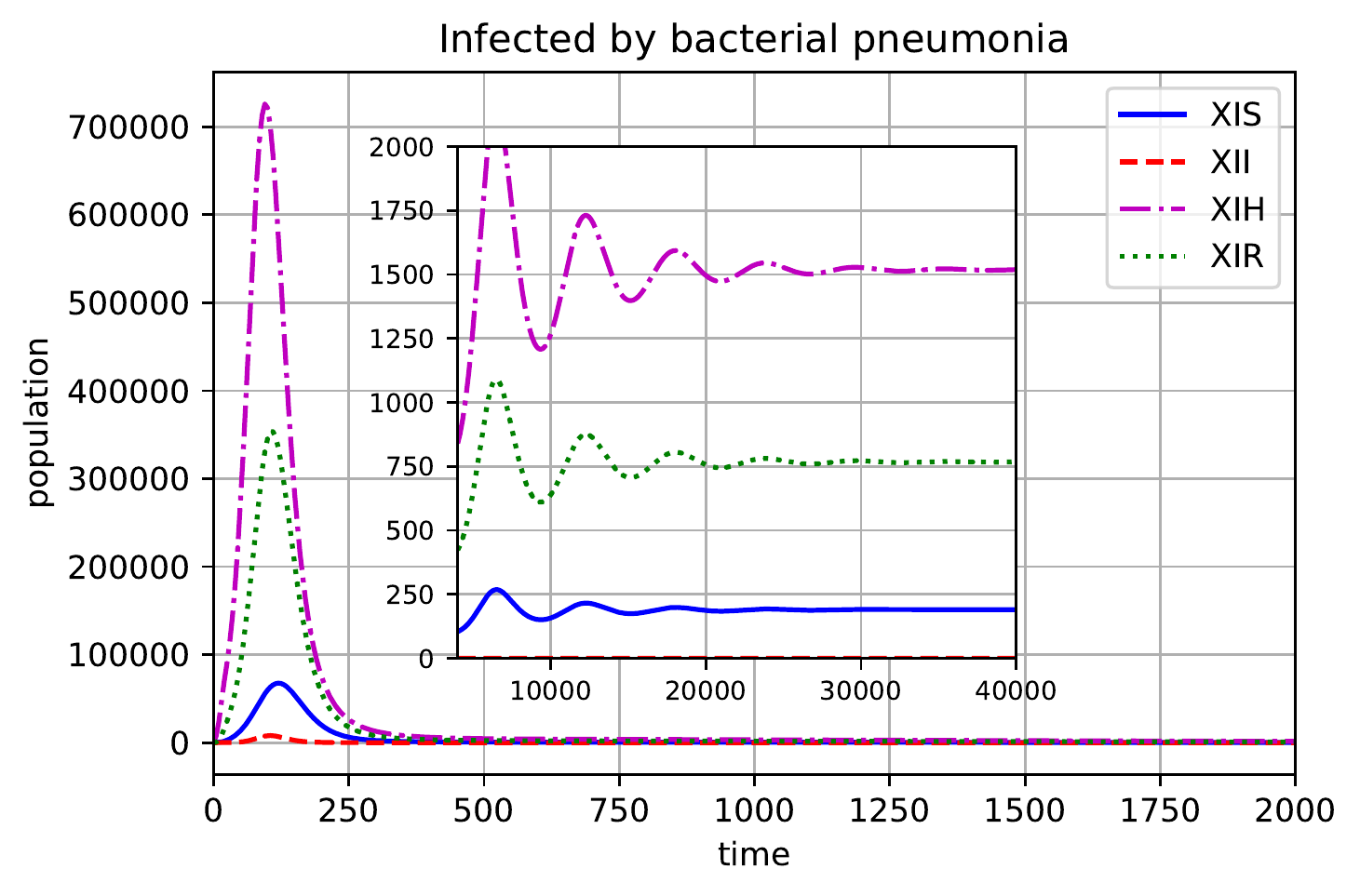}
		\vspace{4ex}
	\end{subfigure} 
	\begin{subfigure}[b]{0.5\linewidth}
		\centering
		\includegraphics[width=0.95\linewidth]{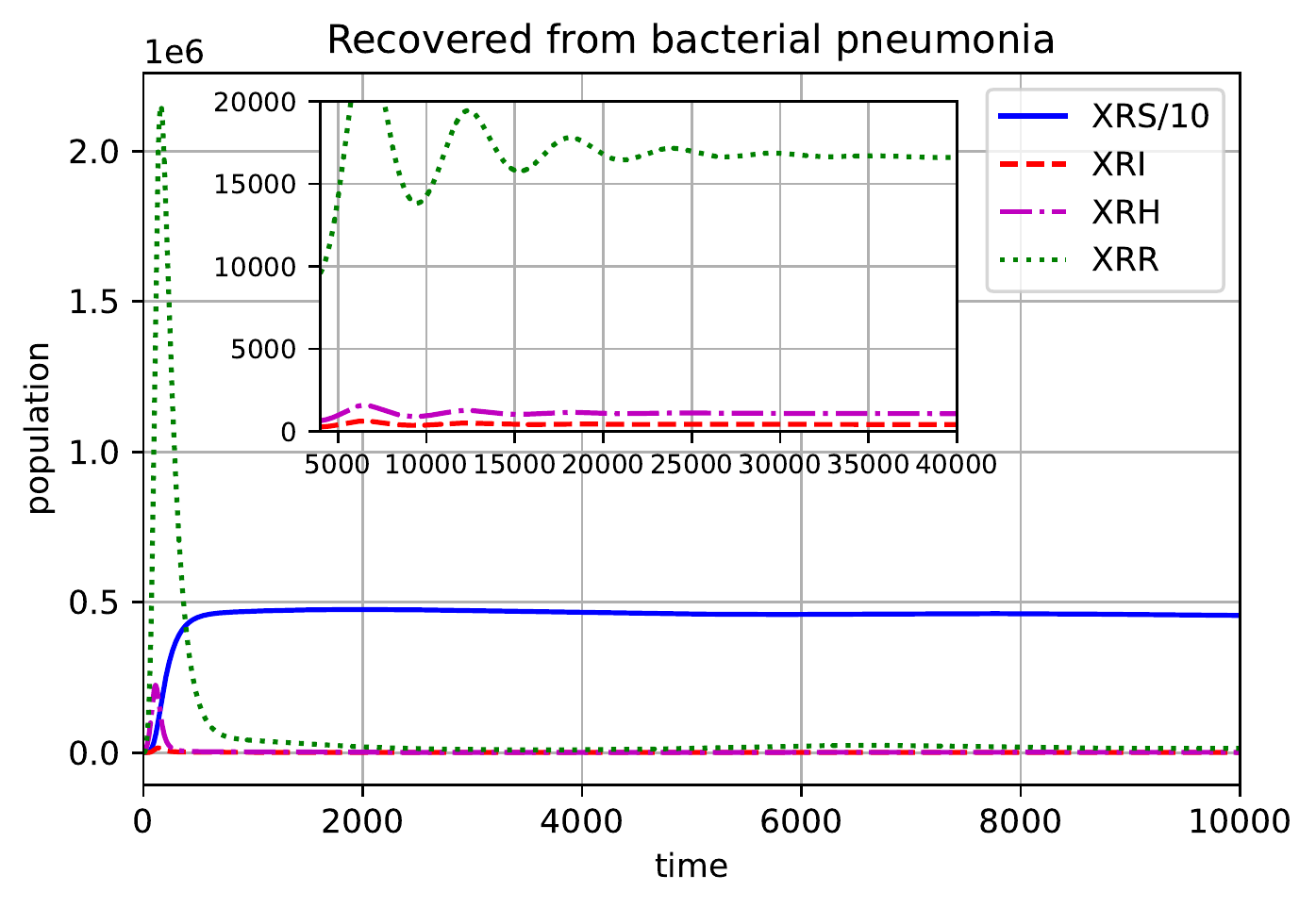}
	\end{subfigure}%%
	\begin{subfigure}[b]{0.5\linewidth}
		\centering
		\includegraphics[width=0.95\linewidth]{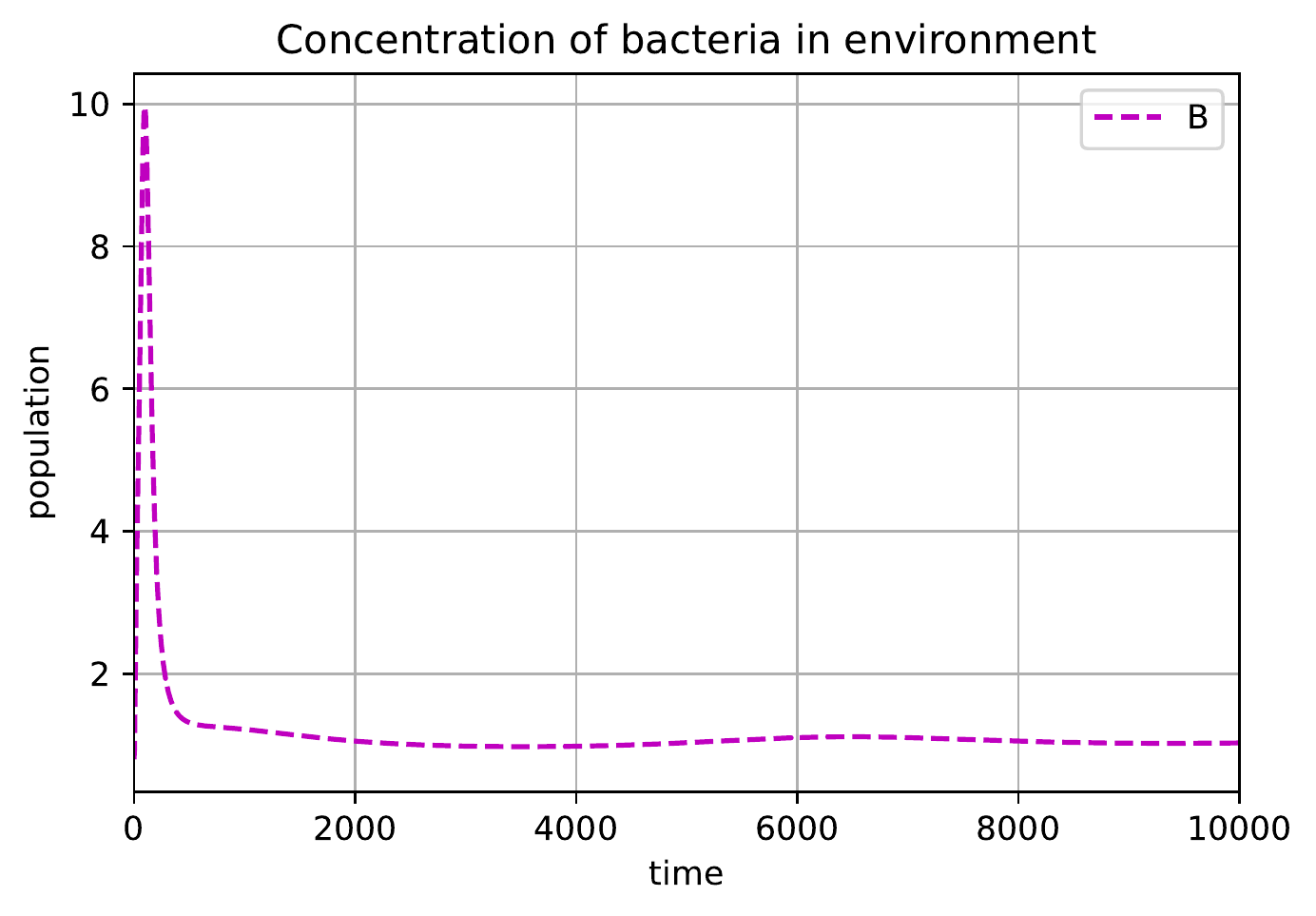}
	\end{subfigure} 
	\caption{Dynamics of the coinfection model when $\mathcal{R}_C>1$, $\mathcal{R}_P<1$ and $\mathcal{R}_B>1$.}
	\label{fig:RcRb} 
\end{figure}

\textbf{Case 4.} When $b=9\times10^{-10}$ and $r=0.08$, we have $\mathcal{R}_P = 1.0352 > 1$ and $\mathcal{R}_B = 8 > 1$. The time plots of the solutions are shown in Figure \ref{fig:RcRpRb}. We can see that the solutions converge to the positive equilibrium
\begin{align*}
	\mathcal{E}_5 & \approx \big(5.728\times10^7,\ 6007,\ 3867,\ 7.747\times10^4,\ 2186,\ 0.618,\ 1467,\ 744.4, \\
	              & \qquad 1.050\times10^7,\ 1102,\ 2074,\ 2.92\times10^4,\ 1.0497\big).
\end{align*}

\begin{figure}
	\begin{subfigure}[b]{0.5\linewidth}
		\centering
		\includegraphics[width=0.95\linewidth]{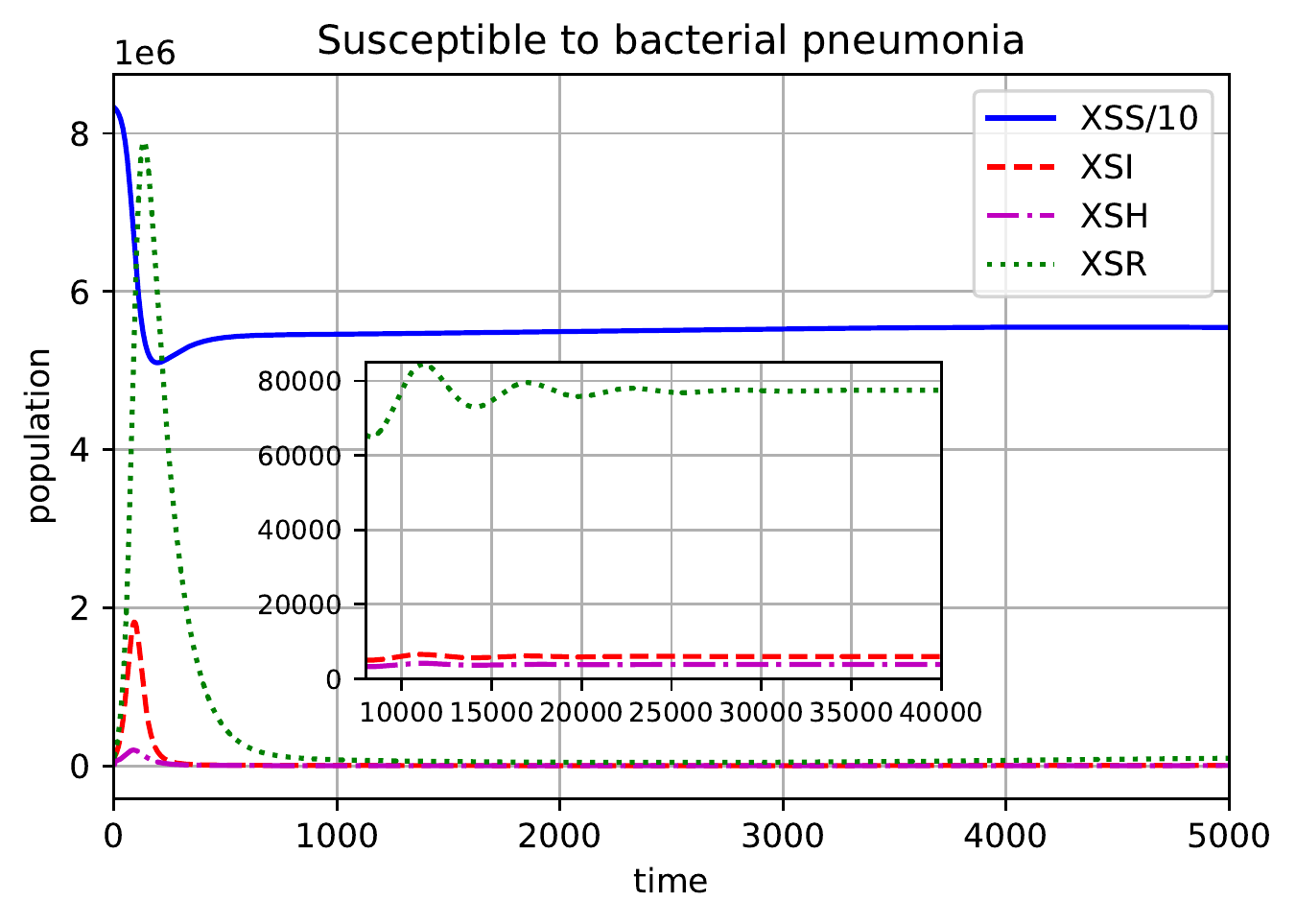}
		\vspace{4ex}
	\end{subfigure}%% 
	\begin{subfigure}[b]{0.5\linewidth}
		\centering
		\includegraphics[width=0.95\linewidth]{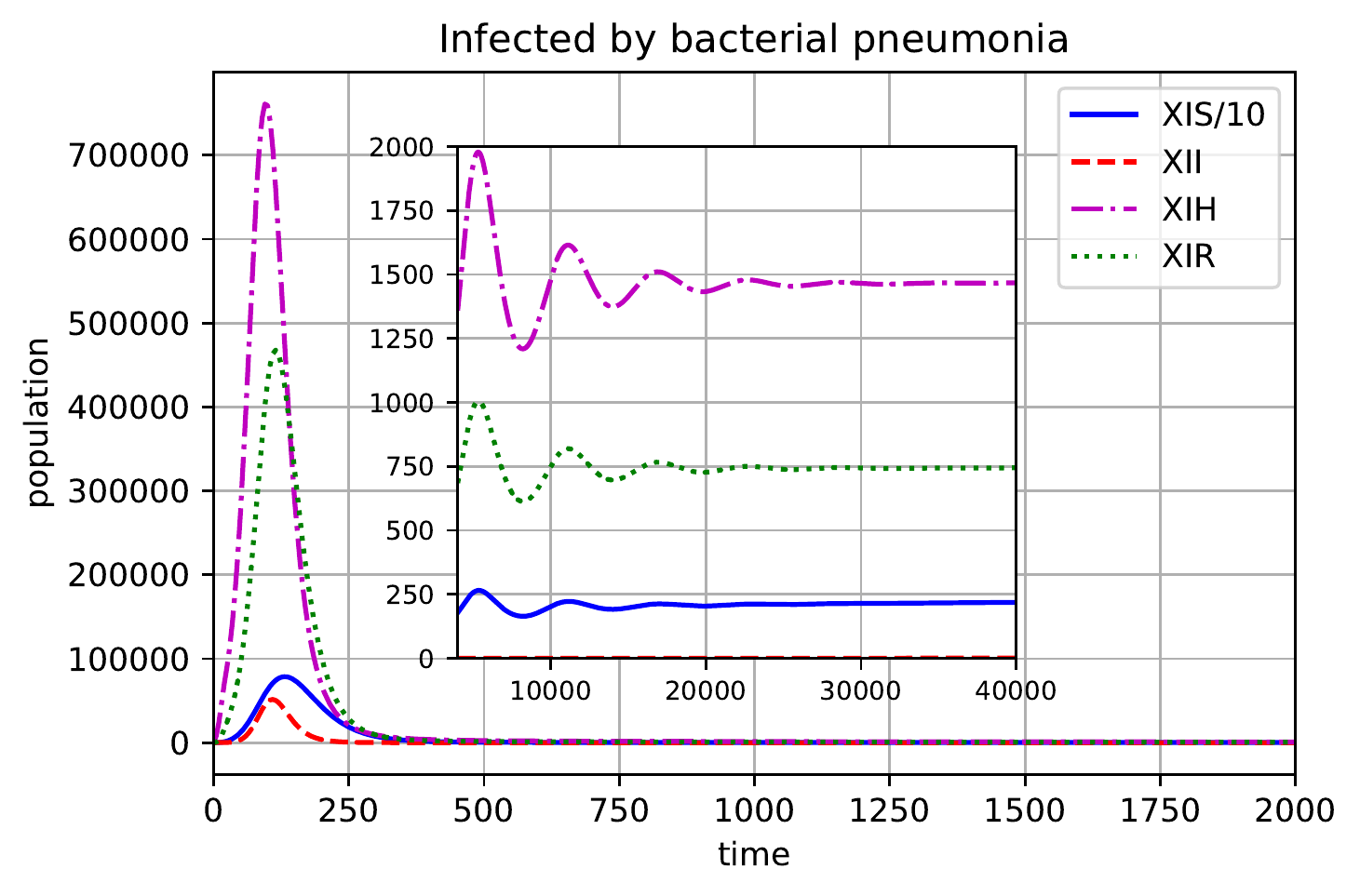}
		\vspace{4ex}
	\end{subfigure} 
	\begin{subfigure}[b]{0.5\linewidth}
		\centering
		\includegraphics[width=0.95\linewidth]{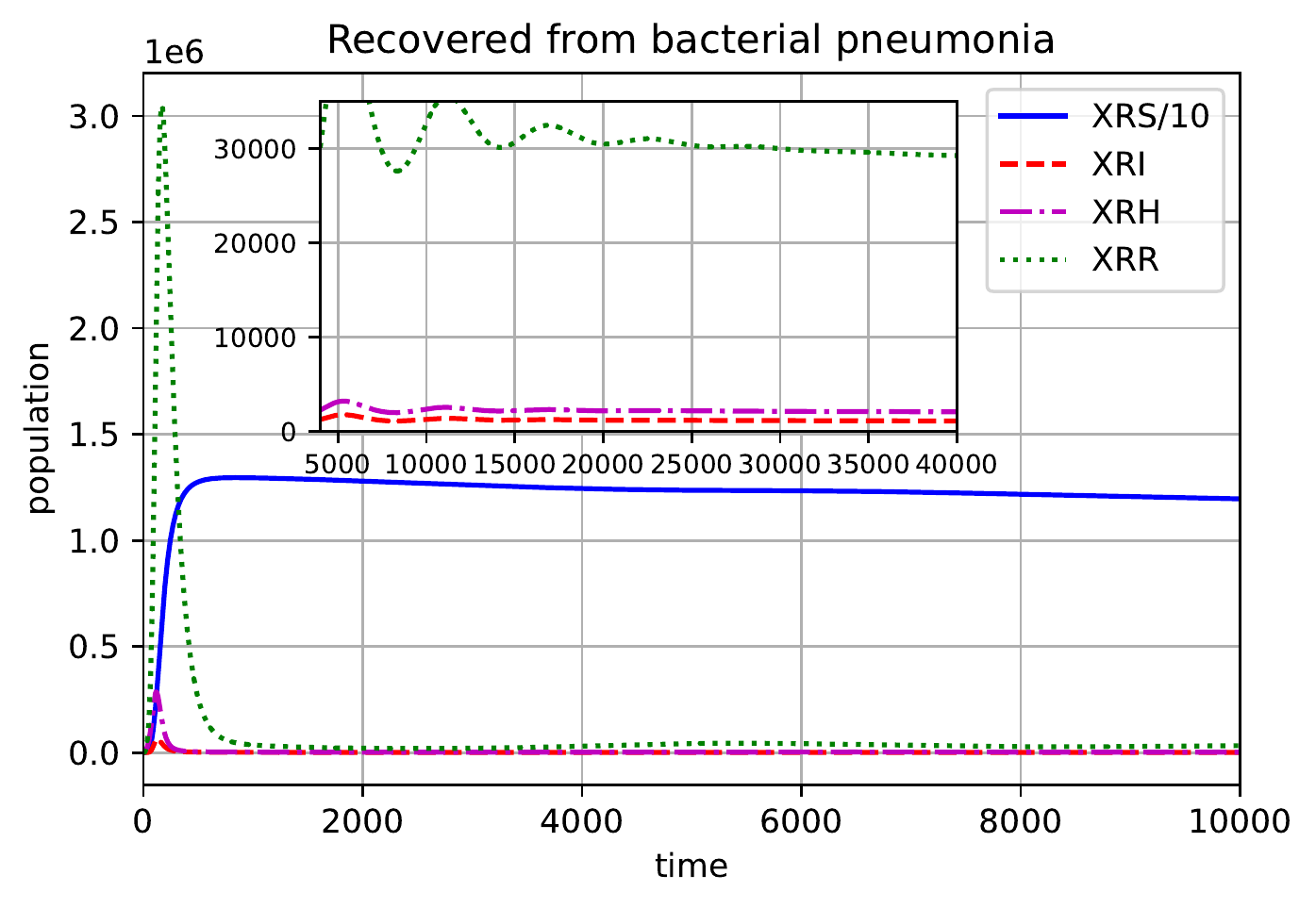}
	\end{subfigure}%%
	\begin{subfigure}[b]{0.5\linewidth}
		\centering
		\includegraphics[width=0.95\linewidth]{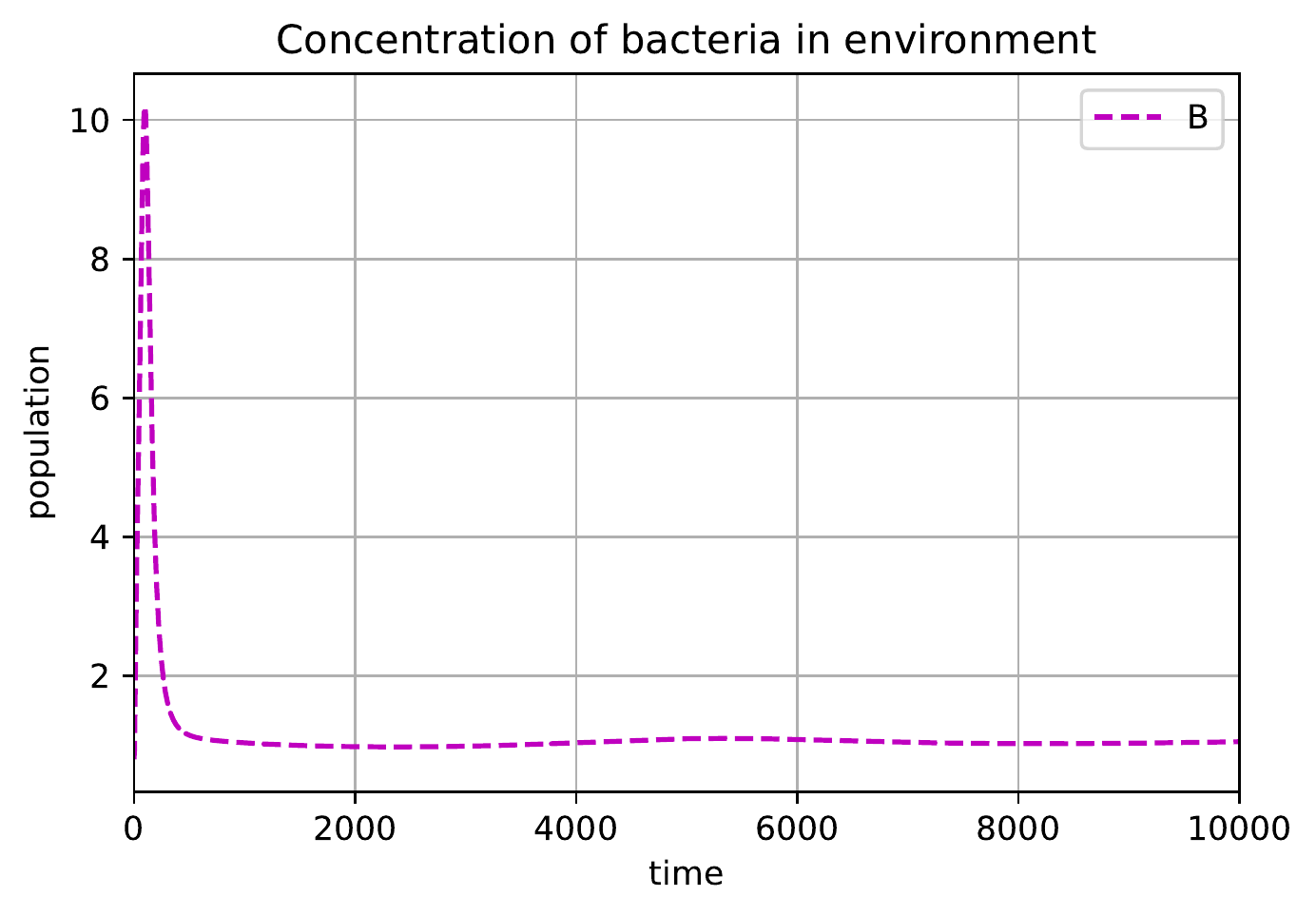}
	\end{subfigure} 
	\caption{Dynamics of the coinfection model when $\mathcal{R}_C>1$, $\mathcal{R}_P>1$ and $\mathcal{R}_B>1$.}
	\label{fig:RcRpRb} 
\end{figure}

\section{Conclusions}
\label{sec:conclusions}

We proposed a novel mathematical model to study the coinfection dynamics of COVID-19 and bacterial pneumonia. We established some basic properties of the sub-models (COVID-19 only and bacterial pneumonia only) and computed their basic reproduction numbers.

We obtained some analytical results for the coinfection model and showed that its dynamics depends on three parameters: $\mathcal{R}_C$, $\mathcal{R}_P$ and $\mathcal{R}_B$. We determined conditions for the existence of five equilibrium points. Furthermore, by means of numerical simulations, we showed that a sixth equilibrium may exist. Based on the simulations on Section \ref{sec:num}, we conjecture that the COVID-19-present, pneumonia-present, bacterial population-present equilibrium $\mathcal{E}_5$ exists and locally stable whenever $\mathcal{R}_C>1$. This implies that both diseases can coexist in the population even if reproduction numbers of bacterial pneumonia ($\mathcal{R}_P$) and bacterial population ($\mathcal{R}_B$) are reduced below unity. Hence, epidemic policies should focus on reducing the basic reproduction number of COVID-19 in order to control the pandemic.

The stability conditions for the equilibria $\mathcal{E}_0$ and $\mathcal{E}_1$ were determined in terms of the reproduction numbers. Due to the complexity of our model, we did not include a stability analysis for all equilibria. On the other hand, the coinfection model could be expanded to include vaccination or multiple COVID-19 variants. We expect to carry out a more thorough analysis in future works.

\section*{Code availability}

The code used in this paper was written in Python and can be downloaded from \url{https://github.com/agcp26/COVID19-pneumonia}.

\printbibliography[heading=bibintoc,title={References}]

\end{document}